\documentclass[11pt]{article}
 
\usepackage[utf8]{inputenc}
\usepackage{authblk}

\usepackage{fullpage}

\usepackage{amsthm}
\usepackage{amsmath}
\usepackage{amssymb}
\usepackage{amsfonts}
\usepackage{mathtools}
\usepackage{enumitem}
\usepackage[pagebackref]{hyperref}
\usepackage[sort,numbers]{natbib}

\usepackage{tabularx}
\usepackage{multirow}
\usepackage{booktabs}
\usepackage{makecell}

\usepackage{algorithm2e}
\DontPrintSemicolon

\usepackage{tikz}

\usepackage[capitalize]{cleveref}

\usepackage[linecolor=green!70!black, backgroundcolor=green!10, bordercolor=black, textsize=scriptsize, obeyFinal
]{todonotes}

\newtheorem{theorem}{Theorem}
\newtheorem{lemma}[theorem]{Lemma}

\newtheorem{corollary}[theorem]{Corollary}

\theoremstyle{definition}
\newtheorem{definition}[theorem]{Definition}

\newtheorem{remark}{Remark}

\newcommand{\threefield}[3]{$#1\mid#2\mid#3$}
\newcommand{\commentout}[1]{}

\DeclarePairedDelimiterX{\abs}[1]{\lvert}{\rvert}{#1}

\DeclareMathOperator{\rep}{rep}

\newcommand{\CES}{\threefield{1}{\rep}{\max_j \sum_i C_{i,j}}\xspace}
\newcommand{\CESbf}{{\boldmath \CES}\xspace}
\newcommand{\LES}{\threefield{1}{\rep}{\max_j \sum_i L_{i,j}}\xspace}
\newcommand{\LESbf}{{\boldmath \LES}\xspace}

\newcommand{\WES}{\threefield{1}{\rep}{\max_j \sum_i W_{i,j}}\xspace}
\newcommand{\WESbf}{{\boldmath \WES}\xspace}

\newcommand{\CESabbrv}{\CES}

\title{Fairness in Repetitive Scheduling}

\author[1]{Danny~Hermelin\thanks{Supported by the ISF, grant No.~1070/20.}} 
\author[2]{Hendrik~Molter\thanks{Supported by the ISF, grants No.~1070/20 and No.~1456/18, and European Research Council, grant number 949707. The main work was done while Hendrik Molter was affiliated with the Department of Industrial Engineering and Management of Ben-Gurion University of the Negev.}}
\author[3]{Rolf~Niedermeier}
\author[4]{Michael~Pinedo}
\author[1]{Dvir~Shabtay$^*$}

\affil[1]{\small Department of Industrial Engineering and Management, Ben-Gurion~University~of~the~Negev, 
Beer-Sheva, 
Israel\\ \texttt{hermelin@bgu.ac.il, dvirs@bgu.ac.il}}

\affil[2]{\small Department of Computer Science, Ben-Gurion~University~of~the~Negev, 
Beer-Sheva, 
Israel\\ \texttt{molterh@post.bgu.ac.il}}

\affil[3]{\small TU Berlin, Faculty IV, Algorithmics and Computational Complexity, Germany\\ \texttt{rolf.niedermeier@tu-berlin.de}}

\affil[4]{\small Leonard N.\ Stern School of Business, New York University, New York, USA\\ \texttt{mlp5@stern.nyu.edu}}

\date{}

\begin{document}

\maketitle

\begin{abstract}
Recent research found that fairness plays a key role in customer satisfaction. Therefore, many manufacturing and services industries have become aware of the need to treat customers fairly. Still, there is a huge lack of models that enable industries to make operational decisions fairly, such as a fair scheduling of the customers' jobs. Our main aim in this research is to provide a unified framework to enable schedulers making fair decisions in repetitive scheduling environments. For doing so, we consider a set of repetitive scheduling problems involving a set of $n$ clients. In each out of~$q$ consecutive operational periods (\emph{e.g.}\ days), each one of the customers submits a job for processing by an operational system. The scheduler's aim is to provide a schedule for each of the $q$ periods such that the quality of service (QoS) received by each of the clients will meet a certain predefined threshold. The QoS of a client may take several different forms, \emph{e.g.}, the number of days that the customer receives its job later than a given due-date, the number of times the customer receive his preferred time slot for service, or the sum of waiting times for service. We analyze the single machine variant of the problem for several different definitions of QoS, and classify the complexity of the corresponding problems using the theories of classical and parameterized complexity. We also study the price of fairness, i.e., the loss in the system's efficiency that results from the need to provide fair solutions. 

\medskip

\end{abstract}

\clearpage

\section{Introduction}
The area of scheduling~\cite{pinedo2012scheduling} is one of the earliest studied fields in operational research and algorithm design. It studies methods by which work is allocated to resources that complete the work. That is, in a typical scheduling problem, we are given a set of jobs that are to be processed on a set of machines so as to optimize a certain scheduling criterion. Such problems occur frequently in numerous industrial and computational settings, and they lie at the core of several efficient implementations of complex systems (\emph{e.g.}\ manufacturing plants, multitask computers, etc.).

In this paper, our aim is to provide a unified framework for making fair decisions in a class of repetitive scheduling problems, where jobs of multiple clients have to be repeatedly scheduled over a finite set of periods. Such scenarios occur frequently both in industrial and commercial settings, as well as in public service settings. Moreover, in many such cases, it is more desirable to ensure customer satisfaction (through some notion of fairness) rather than maximize profits. Our goal is to provide a generic scheme that allows modeling such scheduling scenarios, where the common theme is to ensure that no client receives significantly worse service than any other client over the entire scheduling horizon. 

\subsection{Related work}

Traditionally, in most scheduling models, the objective is to optimize system performance from a global perspective (\emph{e.g.}\ minimizing the time to complete all jobs (makespan) and the total tardiness), ignoring fairness issues. One of the reasons that `fairness’ issues are disregarded in most studies is that there is not one acceptable definition of `fairness'. Whenever definitions of fairness are considered, there is an unending array of philosophic, social, and psychological issues that are immediately brought into consideration, which is precisely the difficulty in arriving at a unified, widely accepted definition~\cite{Wierman}. On top of that, it seems that definitions of fairness which are commonly used in one scheduling domain may not be applicable in other domains, at least not in a direct way.

\textbf{Fairness in classical non-repetitive scheduling.} In classical machine scheduling problems, there are several standard objective functions that address fairness in a very broad sense. For instance, minimizing the maximum lateness of a set of scheduled jobs can be interpreted as a fairness criterion, where we want to make sure that the maximal lateness of any of the jobs is minimized. There are also machine models where fairness comes into play. For example, when scheduling on a set of parallel machines, it is common to define the load on a machine as the total processing time of all jobs assigned to the machine. In such a setting, a fair schedule can be considered as a schedule that balances the load over the machines, so that no machine is assigned too large a load compared to others (see, \emph{e.g.},~\cite{Taub,Azar}). Another class of scheduling problems in which fairness is discussed is that of multi-agent scheduling (see, \emph{e.g.},~\cite{Agnetis2019},~\cite{Stan} and ~\cite{Gerstl}). Our model can be thought of as a generalization of multi-agent scheduling to the repetitive case, where the goal is to compute a set of fair schedules, rather than a single one.

\textbf{Fairness in medical care.} Fairness is a prime consideration in most patient (appointment) scheduling scenarios that are common in medical care. However, even within this specific domain of research, the definition of a \emph{fair schedule} has various different interpretations. For example, in Yan et al.~\cite{Yan} the term fair schedule is defined as one which minimizes the difference between the maximal and minimal waiting time of a patient. On the other hand, in Baum et al.~\cite{Baum} the goal is to find a schedule that maximizes the total revenue of all physicians, while the fairness of a schedule is defined by the minimum revenue of a physician. Note that not only do these papers differ in the way they define fairness, they also differ in their approach towards whom fairness considerations should be applied. %In one application it is the physicians (resources/machines), while in another it is the patients (jobs to process). 
Other examples of `fairness' issues in appointment scheduling problems include Samorani et al.~\cite{Samorani} where fairness is measured by the disparity between the waiting times experienced by the different patients, and Ala et al.~\cite{Ala} where the average waiting time is minimized.

\textbf{Fairness in communication networks.} Fairness issues are also discussed when scheduling bandwidth requests in communication networks (see, \emph{e.g.},~\cite{Bonald},~\cite{Pu}, and~\cite{Balter}). In this domain requests are time-shared. Accordingly, an acceptable definition of a fair schedule is a schedule in which the server assigns its resources fairly to those requests ready to receive service. However, here as well there are several subtleties in the exact notion of fairness that should be considered, and it is not always clear as towards whom fairness considerations should be applied.

\textbf{Fairness in repetitive scheduling systems.}  
The only paper that we found that considers fairness issues in repetitive scheduling problems is that of Ajtai et al.~\cite{Ajtai}. They consider a scheduling problem where each of the clients is also a resource provider (such as in a car sharing problem). Each client that submits a request on a specific day, may be asked to share his resource (car) with other clients submitting requests that same day. The objective is to provide a schedule for car sharing over a finite set of periods (days) such that all clients will have the same relative amount of resource sharing (as much as possible). %We, on the other hand, consider a set of repetitive machine scheduling problems where jobs and resources are different entities in the system. Moreover, in contrast to Ajtai et al.~\cite{Ajtai}, a solution may involve sequencing decisions in addition to job-resource assignment decisions.
A variant of their problem fits nicely within our framework, which naturally allows many more problems to be considered.      

%In~\citet{BKN21}, we presented a special case of our model that considers only the case where the goal is to minimize the number of tardy jobs of each client over the entire scheduling horizon. This paper extends and generalizes the framework we considered in~\cite{BKN21} to include a wide variety of scheduling criteria and machine models. %Moreover, we give a through analysis of three natural objective functions to exemplify the power of our model. 

\subsection{Our contribution}

Our main contribution is a new generic framework that allows the modeling of numerous repetitive scheduling problems where fairness is the main criterion under consideration. This whole area is mostly unexplored, apart from the paper by Ajtai et al.~\cite{Ajtai}, and we believe our framework to be an important first step into exploring the many questions arising in this setting. In~\cite{BKN21}, we presented a special case of our model where the fairness criterion is defined by the number of tardy jobs each client has. The current paper extends and generalizes this framework to include a wide variety of scheduling objectives and machine models. The main advantage of our model is threefold:
\begin{itemize}
\item \emph{Simplicity:} Our framework is very intuitive to understand for people familiar with scheduling problems. In fact, many classical single day scheduling problems, which make sense also in the repetitive setting, easily translate into our model. 
\item \emph{Generality:} Several different job parameters such as due dates and release times, as well various machine models, can all be easily incorporated into our framework. Moreover, jobs constraints such as precedence constraints naturally fit in as well. 
\item \emph{Robustness:} Our framework is easily extendable to handle even more general repetitive scheduling settings, including the case where clients do not submit jobs on certain days, and the case where different notions of fairness have to be applied to each client.
\end{itemize}

To exemplify the power of our model, we give a thorough analysis of three standard objective functions that give rise to three different fairness criteria. Namely, we consider completion time, waiting time, and lateness related fairness criteria (see Section~\ref{sec:model} for formal definitions). Each of these defines a different fair repetitive scheduling problem, and we analyze each of the problems separately. The goal is to explore the computational complexity of each problem, as well as to exhibit the richness of questions that arise from our framework. 

Finally, in the last part of the paper, we introduce the notion of \emph{price of fairness} to the area of repetitive scheduling. This measure, introduced by Bertsimas, Farias, and Trichakis in the context of resource allocation problems~\cite{BFT11}, quantifies the global loss of utility when applying a fair scheduling policy. Thus, analyzing the price of fairness of a certain repetitive scheduling problem allows us to determine the trade-offs between globally optimal and fair schedules.

%\subsection{Preliminaries}
\subsection{Techniques}

To explore the computational complexity landscape of the three examples of our repetitive scheduling problems, we mostly use techniques and notions from standard computational complexity, as can be found in classic textbooks such as~\cite{GareyJS76,Sipser}. We will distinguish between weakly NP-hard problems that are hard only when the instance numbers are huge (super-polynomial in the input length), and strongly NP-hard problems that are hard even if the instance numbers are small (polynomially bounded in the input length). 

We will also use techniques from parameterized complexity~\cite{DowneyFellowsNew,FlumGrohe}. Mostly, we will be concerned with fixed-parameter tractable (FPT) algorithms:
\begin{definition}
An algorithm that solves each instance of a certain problem in $f(k) \cdot n^{O(1)}$ time, where~$n$ is the size of the instance, $k$ the value of the parameter, and $f(k)$ a (monotonically increasing) computable function that depends solely on $k$, is called an FPT algorithm with respect to $k$. 
\end{definition}
Thus, an FPT algorithm implies that a problem is polynomial-time solvable when the given parameter is constant. But in fact, it also allows for polynomial-time algorithms when the parameter (slowly) grows with the input size $n$. The parameters we will consider for our repetitive scheduling problems are: The number of days $q$ in the repetitive scheduling instance, the number of clients $n$ in the instance, and the threshold parameter $K$ (see Section~\ref{sec:model} for more precise definitions).

%\end{multicols}
To exclude the existence of FPT algorithms for some problem with respect to a certain parameter~$k$, it suffices to show that the problem is NP-hard when $k=O(1)$. However, there are several problems that are solvable in polynomial time when their respective parameter is upper bounded by a constant. For such problems, one can show hardness for the parameterized complexity class W[1] to exclude FPT algorithms, which may be thought of as the parameterized analog of showing NP-hardness. The common assumption used in parameterized complexity is then that FPT$\neq$W[1], and this assumption is used to show that certain problems are unlikely to admit FPT algorithms under certain parameterizations. This is done by providing appropriate reductions (called parameterized reductions in the literature) from known W[1]-hard problems. A parameterized reduction is allowed to have running time $f(k) \cdot n^{O(1)}$, where $k$ and $n$ are the input size and parameter value respectively, and its output parameter is required to have a value bounded by some function solely in $k$. We refer to~\cite{DowneyFellowsNew,FlumGrohe} for more information.

\subsection{Paper organization}

The rest of the paper is organized as follows. We formally define our framework for fair repetitive scheduling problems in Section~\ref{sec:model}, and we discuss the three concrete problems considered in the paper. Each one of the following three sections then focuses on one of the three problems we consider. Finally, we give a brief discussion about possible extensions and future work in Section~\ref{sec:conclusion}.

\section{A Model for Fairness in Repetitive Scheduling}
\label{sec:model}

Below we give a formal description of our model in full generality. In what follows we discuss the three repetitive scheduling problems and discuss how they fit within our framework. We close the section by providing an illustrative example of our model, and discuss some further natural applications.

\subsection{Formal model definition}

The class of repetitive scheduling problems we consider in this paper can be defined as follows. We are given a set of $n$ \emph{clients} ($j=1, \ldots, n$), and each client is associated with a single job on each one of $q$ different \emph{days} ($i=1, \ldots, q$). Let $(i,j)$ denote the job of client~$j$ to be scheduled on day $i$. By $\mathcal{J}_i=\{(i,j)\mid j\in \{1,\ldots,n\}\}$, $\mathcal{J}_j=\{(i,j)\mid i\in \{1,...,q\}\}$ and $\mathcal{J}=\{(i,j)\mid i\in \{1,\ldots,q\}, j\in \{1,\ldots,n\}\}$ we denote the set of jobs to be scheduled on day~$i$, the set of jobs that belong to client~$j$, and the set of all jobs, respectively. On each day $i\in \{1,\ldots,q\}$, the jobs in $\mathcal{J}_i$ are to be scheduled non-preemptively on a set of $m$ machines $\mathcal{M}=\{M_{\ell}\mid \ell\in \{1,\ldots,m\}\}$ arranged in a specific and predefined machine environment (such as a single machine (where $m=1$); identical, uniform or unrelated machines in parallel; flow-shop; job-shop or open-shop).

An instance for our problem includes a set of job processing times, $\{p_{i,j,\ell}\mid 1 \leq i \leq  q,\, 1 \leq j \leq n,\, 1 \leq \ell \leq m \}$, where $p_{i,j,\ell}$ represents the \emph{processing time} of job $(i,j)$ on machine $M_\ell$. In addition, when relevant, the instance may include other parameters for each job $(i,j) \in \mathcal{J}$ such as its \emph{due date} $d_{i,j}$, \emph{release date} $r_{i,j}$, and \emph{weight} $w_{i,j}$. For the sake of brevity, we omit the machine (day) index if the scheduling is done on a single machine, or if there is a single day in the scheduling horizon. 

A solution to our problem is a set of $q$ daily schedules $\mathcal{S} = \{\sigma_1, \ldots, \sigma_q\}$. The definition of a feasible daily schedule depends on the specific machine environment (and possibly also on other features and constraints of the problem). Once a solution is defined, we assume that the job completion times can be easily computed. Given a solution, let $C_{i,j}$ be the \emph{completion time} of job $(i,j)$, and let $Z_{i,j}$ %=f(C_{i,j})$\todo{why is it a function of (only) the total completion time? in the case of waiting time or lateness, the function also depends on the processing time / deadline of the job.} 
be the \emph{performance measure} of job $(i,j)$ indicating the quality of service client $j$ receives on day $i$. In such a way, any feasible solution defines a \emph{performance matrix}:
\[
\begin{bmatrix}
Z_{1,1} & Z_{1,2} & \ldots & Z_{1,n} \\
Z_{2.1} & Z_{2,2} & \ldots & Z_{2,n} \\
\vdots & \vdots & \ddots & \vdots \\
Z_{q,1} & Z_{q,2} & \ldots & Z_{q,n}
\end{bmatrix}.
\]
For $1 \leq j \leq n$, we use $Z_j= \sum^q_{i=1} Z_{i,j}$ to denote the total performance measure (or quality of service) of client $j$, which is the sum of the entries of column $j$ in the performance matrix above. Examples of performance measures common in the scheduling literature include: the \emph{completion time} $C_{i,j}$ of job $(i,j)$; the \emph{waiting time} $W_{i,j}=C_{i,j}-p_{i,j}$ of job $(i,j)$; and the \emph{lateness} $L_{i,j}=C_{i,j}-d_{i,j}$ of job $(i,j)$.

We use the three-field notation, \threefield{\alpha}{\beta}{\gamma}, proposed by Graham et al.~\cite{Graham79} to refer to any scheduling problem. The first field ($\alpha$) describes the machine environment. For example, if 1 appears in this field, it means that we deal with a single-machine scheduling problem. The second field ($\beta$) defines the job characteristics and constraints. For example, if $p_{i,j}=p_i$ is specified in the second field, it implies that on each day all processing times are equal (job-independent processing times). We include ``$\rep$" in this field to indicate that we deal with a repetitive scheduling problem.%, and use $\repstar$ to indicate that the same instance repeats itself in each of the $q$ days (i.e., the special case of day-independent instances). 
The last field~($\gamma$) includes information about the scheduling criterion to minimize.

\begin{definition}   
The \threefield{\alpha}{\beta,\rep}{\max_j \sum_i Z_{i,j}} problem asks to determine whether there exists a feasible schedule in which $\max_j \sum_i Z_{i,j} \leq K$, for some prespecified \emph{fairness threshold} (nonnegative integer) $K$. 
\end{definition}

The definition above encapsulates our framework for fair repetitive scheduling. We call any feasible schedule to a \threefield{\alpha}{\beta,\rep}{\max_j {Z_{j}}} problem, where $\max_j Z_j = \max_j \sum_i Z_{i,j} \leq K$, a \emph{$K$-fair} schedule. Such schedules are the main focus of this paper. 
%\todo[inline]{shouldn't we discuss optimization vs.\ decision here? typically one would try to minimize $\max_j {F_{j}}$ in the problem \threefield{\alpha}{\beta,\rep}{\max_j {F_{j}}}. so the canonical decision problem is to ask whether it is possible to get a feasible solution with $\max_j F_j \leq K$.}
 
\subsection{Three specific problems}

In this paper we consider three natural repetitive scheduling settings as examples for our model. In all problems we focus on the single machine unconstrained case, \emph{i.e.}\ the \threefield{1}{\rep}{\max_j {Z_{i,j}}} problem. We consider three natural performance measures $Z_{i,j}$ that are amongst the most classical scheduling criteria studied in the literature. For all these performance measures, a schedule $\sigma_i$ for day $i \in \{1,\ldots,q\}$ is simply a permutation of the $n$ jobs on that day. Thus, a solution for the three problems considered in this paper is simply a set of $q$ permutations.

\CESbf: The first criterion we consider is the total completion time $C_{i,j}$, and the corresponding \CES problem. In this problem, we are given $q$ jobs (one for each day) with processing times $p_{1,j}, p_{2,j},\ldots, p_{q,j}$ for each client $j \in \{1,\ldots,n\}$, and a fairness threshold~$K$. A solution to this problem is a set of $q$ schedules $\mathcal{S}=\{\sigma_1,\ldots,\sigma_q\}$, one for each day.  For such a solution, the completion time $C_{i,j}$ of client~$j$ on day~$i$ is the sum of all processing times $p_{i,\ell}$ such that $\sigma_i(\ell) \leq \sigma_i(j)$, and its total completion time is simply $\sum_i C_{i,j}$, the sum of all completion times on all days. Our goal is, given a threshold parameter $K$, to determine whether there exists a solution where the total completion time of each client is at most $K$. 

\WESbf: The input to the \WES problem is identical to \CES, but we are concerned with waiting times of the clients, rather than their completion times. The waiting time of client $j$ on day $i$ is given by $C_{i,j}-p_{i,j}$, and the total waiting time of client $j$ is $\sum_i W_{i,j}$. Given a threshold parameter $K$, our goal is to determine whether there exists a solution where the total waiting time of each client is at most $K$. 

\LESbf: In the third problem we assume that, in addition to the processing time, each client $j$ has a due date $d_{i,j}$ on day $i$. Given a solution $\mathcal{S}=\{\sigma_1,\ldots,\sigma_q\}$ to a \LES instance, the lateness of client $j$ on day $i$ is $C_{i,j}-d_{i,j}$. Our goal is to determine whether there exists a solution where the total lateness $\sum_i L_{i,j}$ is at most the given fairness threshold $K$.

\begin{remark}
Note that although the problems of minimizing total completion time, waiting time and lateness are equivalent in the classical setting of a single day, as we will see, they are not so in the context of repetitive scheduling. 
\end{remark}

\subsection{An illustrative example}
\label{subsec:example}

Consider three patients, Alice, Bob, and Charlie, of the same doctor. Each of the $n=3$ patients requires a daily treatment in $q=2$ consecutive days. The treatment of Alice requires a single hour each day, the treatment of Bob requires two hours each day, while the treatment of Charlie requires three hours each day. Thus, we have $p_{1,1}=p_{2,1}=1$, $p_{1,2}=p_{2,2}=2$, and $p_{1,3}=p_{2,3}=3$ in this example. Suppose that on each day, all three patients arrive at the doctors office at the same time (say 8:00 a.m.), not knowing in which order the doctor will perform the treatments. The doctor, being considerate of his patients, wishes to focus on the completion time $Z_{i,j}=C_{i,j}$ criteria, \emph{i.e.}\ the time each patient gets to go home.

If the doctor only wishes to compute the \emph{global optima} of this repeated scheduling problem (\emph{i.e.}\ the minimum total completion time throughout all days, $\min \sum_i \sum_j C_{i,j}$), then this corresponds to solving each of the two single day \threefield{1}{}{\sum_j C_j} instances separately. In our example this would correspond to scheduling Alice, Bob, and then Charlie each day, giving Alice a total completion time of $C_1 = C_{1,1} + C_{2,1}=2$, Bob a total completion time of $C_2 = C_{1,2} + C_{2,2}=6$, and Charlie a total completion time of $C_3 = C_{1,3} + C_{2,3}=12$. So 20 hours overall in the doctors office for all patients.    

However, clearly this schedule is not fair to Charlie, as he has to spend a total of 12 hours in the doctors office; twice the amount that Bob does, and six times the amount that Alice does. To be more fair, the doctor could therefore schedule the patients in the reverse order on the second day. This would make Alice spend 7 hours in the office, Bob 8 hours, and Charlie 9 hours, yielding a total of 24 hours. However, no patient needs to spend more than 9 hours at the doctors office in this solution, which is more fair overall, certainly from Charlie's perspective.

Thus, computing the global optima is always equivalent to optimally solving the $q$ single day instances optimally. Furthermore, note that computing the global optima in our framework is not equivalent to computing a $K$-fair solution even for the single day case ($q=1$). For instance in our example of \CES, when the scheduling is done on a single machine and $Z_{j}=C_{j}$, the corresponding global minima problem is solvable in $O(n\log{n})$ time by applying the ``shortest processing time first (SPT)'' rule~\cite{SMITHNS1956}, while any schedule provides a minimal $K$ for which the solution is $K$-fair.

\subsection{Further example applications}

The patient treatment example above occurs in various real life settings, including dialysis treatment, therapeutic procedures, or psychotherapy sessions. In many such settings, the time the treatment ends greatly affects the quality of life of a patient. Thus a fair schedule would try to ensure that the total completion time of treatments of each patient does not exceed a certain threshold. This naturally leads to the \CES problem.

In our example, and in several other medical care applications, it make sense not to penalize the client for having jobs with long processing times. In other settings, that might be unfair. Consider, for example, a set of freelance workers (clients) that apply annually to obtain tax refunds from the government. After the tax files (jobs) have been submitted, they are to be processed by a team of tax officials. Obviously the tax file of each worker is different, and thus the time required to process the requests is client-dependent. Each worker wants to get his refund as soon as possible to improve his cash flow. In such a case, a fair solution would be to minimize the total waiting time of the workers over a number of years, since workers with simple files should not be penalized for other workers with complicated files. Thus, we get a natural application of the \WES problem.

%Consider first the following patient scheduling problem: There is a group of patients that require treatment on a single machine (e.g., dialysis machine, or cancer radiation machine). The machine works on a daily basis starting at 7:00 AM and each patient who needs treatment wants to be scheduled as early as possible, so he will be able to continue his regular daily schedule (attend work, for example) as early as possible. To be fair, the scheduler may want to generate $q$ schedules, such that the average completion time of each of the patients is not larger than $K$ (which is the decision version of our maximal fairness problem). In another application, we may assume that each patient has a due date in each of the days which is the last train to his work (or home), or a due-date given by his employee to arrive at work. In such a case the scheduler may want to be fair by making sure that each patient will not be late in more than $K$ days. Lastly, in many cases patients prefer to be scheduled either early in the morning (at the beginning of the working day) or late in the afternoon (or even evening) so that the working day will not be affected. In such case, it is common to use a unimodal performance measure function to evaluate the quality of service that clients receive.% (see Section $\ref{sec:Def}$).    

Finally, consider custom checks on imported containers in a port. There are $n$ companies, and each has its own set of containers that have to be checked by the customs team every day. The team inspects (either by sampling or by checking all containers) each one of the companies, one after the other. Each company has a due date dictated by its customers requirements. The customers are willing to pay a bonus for early delivery, but demand that the company pays a penalty for late delivery, each depending on how early or late the delivery is. In such a case, a fair scheduler may want to minimize the total lateness over all days of each client, \emph{i.e.}\ the \LES problem.

\subsection{Possible extensions}

Our framework is robust in a sense that it can easily be extended to even more general repetitive settings. Below we discuss two possible extensions.

\begin{itemize}

\item In some cases, we might want to have different fairness thresholds for different clients (\emph{e.g.}\ elderly versus young patients). This can easily be modeled in our framework by assigning client-specific fairness thresholds $K_j$. In fact, the \LES problem reduces to such a generalization of \CES. 
\item In some applications, a client may submit a job only in a subset of the $q$ days. In such cases, an immediate extension of our model would be to redefine $Z_j$ to be the average performance measure of client $j$ over all days in which he submitted a job, \emph{i.e.}\ $Z_j = 1/q_j\sum_{i=1}^q Z_{i,j}$, where~$q_j$ is the number of days in which client $j$ submitted a job, and $Z_{i,j}=0$ if client $j$ did not submit a job on day $i$. %It is interesting to see which of our results extends to this case.
%\todo[inline]{i think the second case is a special case of the first one. the second one can be modeled by setting $p_{i,j}=0$ if client $j$ does not submit a job on day $i$ (which yields w.l.o.g.\ that $F_{i,j}=0$ at least for all performance measures that we consider) and then instead of requiring $1/q_j\sum_{i=1}^q F_{i,j}\le K$ we require $F_j = \sum_{i=1}^q F_{i,j}\le K_j=K\cdot q_j$.}
\end{itemize}

%%%%%%%%%%%%%%%%%%%%%%%%%%%%%%%%%%%%%%%%%%%%%%%%%%%%%%%%%%%%%%%%
%%%% Completion time W_{i,j}
%%%%%%%%%%%%%%%%%%%%%%%%%%%%%%%%%%%%%%%%%%%%%%%%%%%%%%%%%%%%%%%%

\section{The \CESbf Problem}
\label{sec:CES}%

As a first example of our framework, we consider the \CES problem. Recall that in this problem, we want to determine whether there exists a set of $q$ schedules such that the total completion time $\sum_i C_{i,j}$ of each client $j$ over all $q$ days is at most $K$ (the fairness threshold). We present a thorough analysis of the computational complexity landscape of the problem with respect to parameters $n$ (number of clients) and $q$ (number of days), considering both the cases where the processing times are allowed to be huge or are encoded in unary (and so are polynomially bounded by the input size).

\subsection{A polynomial-time solvable special case}

We start with the special case of \CES with only two days ($q=2$). We show that this case is polynomial-time solvable, based on an observation regarding $K$-fair schedules for this special case. Namely, we show that if there exists a $K$-fair solution, then there also exists a $K$-fair solution which schedules the jobs on day 2 in the reverse order of day 1. Later on, we show how to use this property to construct a polynomial-time algorithm.

\begin{lemma}%[\appref{lemma:reverse}]
\label{lemma:reverse}
For any $K$, there exists a $K$-fair schedule for a \CES instance with $q=2$ if and only if there exists a $K$-fair schedule where the schedule for day 1 is in the reverse order of the schedule of day 2.
\end{lemma}
%\appendixproof{lemma:reverse}{
\begin{proof}
Assume that there is a $K$-fair solution to a \CESabbrv instance with $q=2$, and there exist two clients $j$ and $j'$ such that the jobs of client $j$ are scheduled after the jobs of client~$j'$ on both days. Then moving the job of client $j'$ directly after the job of client~$j$ on one of the days decreases the total completion time of~$j$, while the total completion time of client $j'$ after the modification is still smaller than the total completion time of client~$j$ before the modification. Hence, the new schedule is $K$-fair as well. We can repeat this procedure until there are no more pairs of clients such that one of them has its jobs scheduled after the other one on both days.
\end{proof}
%}

Using \cref{lemma:reverse}, we design an algorithm that is similar to Johnson's algorithm~\cite{pinedo2012scheduling} for minimizing the makespan in a two-machine flow-shop scheduling problem (\emph{i.e.}\ the \threefield{F2}{}{C_{\max}} problem). Similar in the sense that it uses the same order of Johnson's algorithm on the first day, but a different one for the second day because of Lemma~\ref{lemma:reverse}.

\begin{theorem}
\label{thm:2days}%
The \CESabbrv problem with $q=2$ is solvable in $O(n\log n)$ time.
\end{theorem}
\begin{proof}
We call clients $j$ with $p_{1,j}\le p_{2,j}$ \emph{Type~1} clients (no larger processing time on day~1), and clients $j$ with $p_{1,j}> p_{2,j}$ \emph{Type~2} clients (smaller processing time on day~2). 
We create a schedule for day 1 that starts with Type~1 jobs in non-decreasing order of~$p_{1,j}$ and then Type~2 jobs in non-increasing order of~$p_{2,j}$. Such schedules are referred to as \emph{SPT(1)-LPT(2)} schedules~\cite{pinedo2012scheduling} which are known to be optimal for \threefield{F2}{}{C_{\max}}~\cite{pinedo2012scheduling} (flow shop scheduling with two machines and makespan minimization). We use the reverse order on day~2. This is a major difference from the algorithm  for \threefield{F2}{}{C_{\max}}, where the same order is used for the second machine. It is straightforward to check that the algorithm obeys the claimed running time bound, since the sorting step dominates the overall running time. We now show the correctness of the algorithm. 

Consider an optimal set of schedules $\mathcal{S}^\star=\{\sigma^\star_1,\sigma^\star_2\}$. By \cref{lemma:reverse} we know that there is an optimal schedule where day 1 is the reverse of day 2, so assume w.l.o.g.\ that $\sigma^\star_1$ is the reverse of $\sigma^\star_2$.
 
Assume that there is at least one client of Type~1 and at least one client of Type~2. We first show that there is an optimal set of schedules where the jobs of all Type~1 clients are scheduled earlier than the jobs of Type~2 clients on day 1. Take an optimal set of schedules $\mathcal{S}^\star=\{\sigma^\star_1,\sigma^\star_2\}$ where $\sigma^\star_1$ is the reverse of $\sigma^\star_2$ and let client $j_1$ be the Type~1 client with the earliest scheduled job on day 1 such that there is a Type~2 client that has a job scheduled earlier on day 1. Note that this implies that a job of a Type~2 client $j_2$ is scheduled directly before the job of client $j_1$. It is easy to observe that swapping the jobs of these two clients on both days does not affect the total completion time of all other clients. Furthermore, the total completion time of each of these two clients is also not increased by this swap: Let~$\mathcal{S}^{\star\star}$ denote the set of schedules obtained from $\mathcal{S}^\star$ after swapping the jobs of clients~$j_1$ and~$j_2$ on both days. Then $C_{j_1}(\mathcal{S}^{\star\star})=C_{j_1}(\mathcal{S}^{\star})-p_{1,j_2}+p_{2,j_2}\le C_{j_1}(\mathcal{S}^{\star})$, since $p_{1,j_2}>p_{2,j_2}$. Analogously $C_{j_2}(\mathcal{S}^{\star\star})=C_{j_2}(\mathcal{S}^{\star})+p_{1,j_1}-p_{2,j_1}\le C_{j_2}(\mathcal{S}^{\star})$, since $p_{1,j_1}\le p_{2,j_1}$. Hence, we can iteratively swap Type~1 and Type~2 clients until all Type~1 clients are earlier on day 1 than the Type~2 clients.
 
Now assume there are two Type~1 clients $j_1$ and $j_2$ such that $p_{1,j_1}\le p_{1,j_2}$ but the job of~$j_1$ is scheduled after the job of~$j_2$ on day 1. W.l.o.g.\ we can assume that $j_1$'s job is scheduled \emph{directly} after the job of $j_2$ on day 1. Again we have that swapping the jobs of these two clients on both days does not affect the total completion time of all other clients. Let~$\mathcal{S}^{\star\star}$ denote the set of schedules obtained from $\mathcal{S}^\star$ after swapping the jobs of clients~$j_1$ and~$j_2$ on both days. Note that the completion time of~$j_1$ on day 2 after the swap is the same as the completion time of $j_2$ on day 2 before the swap. The completion time of $j_1$ on day 1 after the swap is the completion time of $j_2$ on day 1 before the swap plus $p_{1,j_1}-p_{1,j_2}\le 0$. Hence, we have $C_{j_1}(\mathcal{S}^{\star\star})=C_{j_2}(\mathcal{S}^{\star})-p_{1,j_2}+p_{1,j_1}\le C_{j_2}(\mathcal{S}^{\star})$. 
Furthermore, we can write $C_{j_2}(\mathcal{S}^{\star\star})=C_{j_2}(\mathcal{S}^{\star})+p_{1,j_1}-p_{2,j_1}\le C_{j_2}(\mathcal{S}^{\star})$. Hence, we have that the total completion time of both clients after the swap equals at most the total completion time of $j_2$ before the swap. This implies that $\mathcal{S}^{\star\star}$ is at least as good as $\mathcal{S}^{\star}$. 
  
The case where we have two Type~2 clients $j_1$ and $j_2$ such that $p_{2,j_1}\ge p_{2,j_2}$ but the job of~$j_1$ is scheduled after the job of~$j_2$ on day 1 is completely analogous. The two days change their role.
\end{proof}

\subsection{Two weakly NP-hardness results}

Recall that by Theorem~\ref{thm:2days}, the problem \CES is polynomial-time solvable in case there are at most two days in the scheduling horizon. We next complement this result by showing that the problem is (weakly) NP-hard when there are at least four days in the horizon. This leaves only the case of $q=3$ open. 

Our hardness result is based on a reduction from the well-known weakly NP-hard \textsc{Partition} problem~\cite{GJ79}.
\begin{definition}
\label{def:partition}
\textsc{Partition}: Given a set $A=\{a_1, \ldots, a_{|A|}\}$ of integers with $\sum_{a\in A}a=2B$, is there a partition of $A$ into $A_1$ and $A_2$ $($i.e., $A_1\cup A_2=A$ and $A_1\cap A_2=\emptyset$$)$ such that $\sum_{a\in A_1}a=\sum_{a\in A_2}a=B$?
\end{definition}

\begin{theorem}%[\appref{thm:fewdays}]
\label{thm:fewdaysNPh}
The \CES problem is weakly NP-hard for all $q\ge 4$. 
\end{theorem}
%\appendixproof{thm:fewdays}{
\begin{proof}
Given an instance $I$ of \textsc{Partition}, we construct an instance for the \CES problem with $n=|A|+2$ and $q=4$ in the following way. 
%We later explain how to remove zero processing times.
\begin{itemize}
    \item We create one client $x$ with $p_{1,x}=2B$, $p_{2,x}=5B$, and $p_{3,x}=p_{4,x}=0$.
    \item We create one client $y$ with $p_{1,y}=p_{2,y}=0$, and $p_{3,y}=2B$, $p_{4,y}=5B$.
    \item For each $a_j\in A$ we create one client $j$ with $p_{2,j}=p_{4,j}=a_j$ and $p_{1,j}=p_{3,j}=0$.
    \item We set $K=8B$.
\end{itemize}
%Note that in any yes-instance of \textsc{Partition}, the sum of all integers in $S$ needs to be an even number, hence we have that all processing times are integer. 
To obtain the result for more than four days, we can add additional days where all clients have jobs with processing time zero.
The described \CES instance can clearly be computed in polynomial time.

\emph{Correctness.} We now show that the \textsc{Partition} instance is a yes-instance if and only if the constructed \CES instance is a yes-instance.

($\Rightarrow$): Assume that the given \textsc{Partition} instance is a yes-instance. Then there exists a partition of $A$ into $A_1$ and $A_2$ such that $\sum_{a_j\in A_1}a_j=\sum_{a_j\in A_2}a_j=B$. Then we schedule the jobs in the following way. 
\begin{itemize}
    \item On day 1, we schedule the job of client $x$.
    \item On day 2, we first schedule the jobs of clients $j$ with $a_j\in A_1$ in an arbitrary order. Then we schedule the job of client $x$. Afterwards, we schedule the jobs of clients $j$ with $a_j\in A_2$ in an arbitrary order.
    \item On day 3, we schedule the job of client $y$.
    \item On day 4, we first schedule the jobs of clients $j$ with $a_j\in A_2$ in an arbitrary order. Then we schedule the job of client $y$. Afterwards, we schedule the jobs of clients $j$ with $a_j\in A_1$ in an arbitrary order.
\end{itemize}

Note that above we specified how to schedule all jobs with non-zero processing times. All remaining jobs are scheduled first on their respective days in arbitrary order.

The total completion times of clients $x$ and $y$ are each exactly $K$, since we know that $\sum_{a_j\in A_1}a_j=\sum_{a_j\in A_2}a_j = B$. Now, assume that there is a client~$j'$ with $a_{j'}\in A_1$ that is scheduled last among all clients~$j$ with $a_j\in A_1$ on both days 2 and 4. The total completion time of client~$j'$ is then also exactly $K$. It follows that the total completion times of all other clients~$j$ with $a_j\in A_1$ are all at most $K$. By symmetry, we also have that the total completion time of all clients $j$ with $a_j\in A_2$ is at most $K$.

($\Leftarrow$): Assume that the constructed \CESabbrv instance is a yes-instance and we have a $K$-fair set of schedules $\mathcal{S}$. Consider day 2. For all clients $j$ whose jobs are scheduled earlier than the job of client $x$ on day 2, we put $a_j$ into $A_1$. For all remaining clients~$j$, we put $a_j$ into $A_2$. Assume for contradiction that $\sum_{a_j\in A_1}a_j\neq\sum_{a_j\in A_2}a_j$. Consider first the case $\sum_{a_j\in A_1}a_j>\sum_{a_j\in A_2}a_j$. Then $C_x(\mathcal{S})=7B+\sum_{a\in A_1}a>8B=K$, a contradiction. Now consider the case $\sum_{a_j\in A_1}a_j<\sum_{a_j\in A_2}a_j$. Assume for contradiction that all clients~$j$ with $a_j\in A_2$ have their job scheduled earlier than client~$y$ on day 4. Then by a similar argument as above we get that $C_y(\mathcal{S})>8B=K$, a contradiction. Hence we have that there is a client $j$ that has its job scheduled after client $x$ on day 2 and after client $y$ on day 4. However, then we have that $C_j(\mathcal{S})>10B>K$, a contradiction.
%
%\todo[inline]{todo: show how to get rid of zero processing times. idea: show that you can replace them something very small and increase $k$ by something sufficiently small. then rescale everything to make all numbers integer.}
\end{proof}
%}

Theorems~\ref{thm:2days} and~\ref{thm:fewdaysNPh} together draw a clear picture of the NP-hardness landscape of \CES with respect to parameter~$q$, the number of days in the scheduling horizon. The only case that remains open is $q=3$. We next consider parameter~$n$, the number of clients. We show that the problem is (weakly) NP-hard already for $n=2$ (but $q$ being arbitrary).

\begin{theorem}%[\appref{thm:fewclientsNPh}]
\label{thm:fewclientsNPh}
The \CES problem is weakly NP-hard for all $n\ge 2$. 
\end{theorem}
%\appendixproof{thm:fewclientsNPh}{
\begin{proof}
We present a polynomial-time reduction from \textsc{Partition} (see Definition~\ref{def:partition}). Given an instance $I$ of \textsc{Partition}, we construct an instance of \CESabbrv with $n=2$ clients and $q=|A|$ days in the following way. We set
\[
p_{i,j} = a_i \text{ for all } j\in\{1,2\} \text{ and } i\in\{1,\ldots,|A|\},
\]
and we set $K=3B$. To obtain the result for more than two clients, we can add additional clients that have jobs with processing time zero.
This can clearly be done in polynomial time.

\emph{Correctness.} We now show that the \textsc{Partition} instance is a yes-instance if and only if the constructed \CESabbrv instance is a yes-instance.

($\Rightarrow$): Assume that the given \textsc{Partition} instance is a yes-instance. Then there exists a partition of $A$ into $A_1$ and $A_2$ such that $\sum_{a_j\in A_1}a_j=\sum_{a_j\in A_2}a_j=B$. We schedule the jobs in the following way: For all days $i\in\{1,\ldots, |A|\}$ if $a_i\in A_1$, then we schedule the job of client 1 first and the job of client 2 second on day~$i$. Otherwise, we schedule the job of client 2 first and the job of client 1 second on day~$i$. Observe that for each $j \in \{1,2\}$ and $i\in\{1,\ldots, |A|\}$ we have that $C_{i,j}=a_i$ if and only if $a_i\in A_j$, and $C_{i,j}=2a_i$ otherwise. Now it is easy to check that the total completion time for client~1 is $\sum_{a_j\in A_1}a_j + 2\sum_{a_j\in A_2}a_j= 3B$. Analogously, we have that the total completion time of client 2 is $\sum_{a_j\in A_2}a_j+2\sum_{a_j\in A_1}a_j=3B=K$.

($\Leftarrow$): Assume that the constructed \CESabbrv instance is a yes-instance. Then there exists a set $\mathcal{S}$ of $q$ schedules $\{\sigma_1,\ldots,\sigma_q\}$ so that for each $j \in \{1,2\}$ we have $\sum_{i=1}^q C_{i,j}\le \frac{3}{2}\sum_{a_j\in A}a_j$. We claim that $A_1=\{a_i \mid \sigma_i(1)=1\}$ and $A_2=\{a_i \mid \sigma_i(2)=1\}$ constitutes a solution for the given \textsc{Partition} instance. Suppose for contradiction that it is not. First, observe that by construction $A_1\cup A_2=A$ and $A_1\cap A_2=\emptyset$, that is, we have a partition of $A$. Now if the partition of $A$ into $A_1$ and $A_2$ is not a solution for the \textsc{Partition} problem, then w.l.o.g.\ we have that $\sum_{a_j\in A_1}a_j >\frac{1}{2}\sum_{a_j\in A}a_j > \sum_{a\in A_2}a$. Now consider the total completion time of client 2. It is $\sum_{a_j\in A_2}a_j + 2\sum_{a_j\in A_1}a_j = \sum_{a_j\in A} a_j+ \sum_{a_j\in A_1} a_j > 3B=K$, a contradiction.
\end{proof}

\subsection{Strong NP-hardness}

Theorems~\ref{thm:2days}, \ref{thm:fewdaysNPh}, and~\ref{thm:fewclientsNPh} give a complete picture of the complexity of \CES in terms of $n$ and $q$, apart from the case of $q=3$. However, in all results above we assumed that the processing times are allowed to be huge (exponential in the input size). What happens when the processing times of the jobs are all small (polynomial bounded in the input size)? Below we show that unfortunately \CES remains NP-hard not only when all processing times are constant, but even when the fairness threshold is constant. 

Our proof is done by a polynomial-time reduction from \textsc{Exact (3,4)-SAT}, which is known to be NP-hard~\cite{Tov84}.

\begin{definition}
\textsc{Exact (3,4)-SAT}: Given a Boolean formula $\phi$ (in a conjunctive normal form), where every clause contains exactly three distinct variables and every variable appears in exactly four clauses, we need to identify whether there is a satisfying assignment for $\phi$.
\end{definition}

\begin{theorem}
\label{thm:strNPh}
The \CES problem is strongly NP-hard even if $K\in O(1)$.
\end{theorem}
\begin{proof}
Given an instance for the \textsc{Exact (3,4)-SAT}, we construct an instance of \CES in the following way. 
\begin{itemize}
\item For each variable $x$ that appears in $\phi$, we create three clients $x, x^T, x^F$ and two days~$x_1,x_2$. We set processing times $p_{x_1,x}=9$, $p_{x_2,x}=27$. We set $p_{x_2,x^T}=p_{x_2,x^F}=1$.
\item For each clause $c$ that appears in $\phi$, we create one client $c$ and two days $c_1,c_2$. We set $p_{c_1,c}=29$, $p_{c_2,c}=6$. Let $x,y,z$ be the three variables appearing in $c$. If $x$ appears non-negated, then we set $p_{c_2,x^T}=1$, otherwise we set $p_{c_2,x^F}=1$. Analogously for $y$ and~$z$.
\item For each combination of a client $j$ and a day $i$ where we have not specified $p_{i,j}$ above, we set $p_{i,j}=0$.
\item We set $K=37$.
\end{itemize}
This finishes the construction of the \CES instance and can clearly be done in polynomial time. 

\emph{Correctness.} We now show that $\phi$ is satisfiable if and only if the constructed \CES instance is a yes-instance.

($\Rightarrow$): Assume that $\phi$ is satisfiable and we have a satisfying assignment for the variables in~$\phi$. We construct a $K$-fair schedule as follows.
\begin{itemize}
\item For each variable $x$ we schedule the job of client $x$ first on day $x_1$. If $x$ is set to \texttt{true} in the satisfying assignment, then on day $x_2$ we schedule the jobs of client $x_T$ first, the job of client~$x$ second, and the job of client $x^F$ third. If $x$ is set to \texttt{false} in the satisfying assignment, then on day $x_2$ we schedule the jobs of client $x_F$ first, the job of client~$x$ second, and the job of client $x^T$ third.
\item For each clause $c$ we schedule the job of client $c$ first on day $c_1$. Let $x,y,z$ be the variables in $c$ and let w.l.o.g.\ $x$ be a variable that is set to a truth-value that satisfies $c$. We schedule on day~$c_2$ the jobs of client $y^T$ (resp.\ $y^F$) first, the job of client $z^T$ (reps.\ $z^F$) second, the job of client~$c$ third, and the job of client $x^T$ (resp.\ $x^F$) fourth, depending on whether the variables appear negated or non-negated in $c$ (say, $x$ appears non-negated, then we schedule~$x^T$, otherwise we schedule $x^F$).
\item Note that above we specified how to schedule all jobs with non-zero processing times. All remaining jobs are scheduled first (that is, before any job with non-zero processing time) on their respective days in arbitrary order.
\end{itemize}
Call the above set of schedules $\mathcal{S}$. Assume for contradiction that $\mathcal{S}$ is not a $K$-fair set of schedules. Then there is a client $j$ with $C_j(\mathcal{S})>K$. It is easy to observe that for all variables $x$ we have that $C_x(\mathcal{S})=K$ and for all clauses $c$ we have that $C_c(\mathcal{S})=K$. It follows that $j=x^T$ or $j=x^F$ for some variable $x$.

Assume w.l.o.g\ that $j=x^T$ (the case that $j=x^F$ is analogous). We know that there are four clauses $a,b,c,d$ in $\phi$ that contain variable $x$. Note that for $i\notin\{x_2,a_2,b_2,c_2,d_2\}$ we have that $p_{i,x^T}= 0$. We claim that $C_{x^T}(\mathcal{S})>K$ implies that the job of client $x^T$ is scheduled last on day~$x_2$ and on at least one of the days $a_2,b_2,c_2,d_2$. Assume that the job of client~$x^T$ is scheduled first on day $x_2$. Note that the total processing time of all jobs on each of the days $a_2,b_2,c_2,d_2$ is $9$. This means that the total completion time of client $x^T$ is at most~$K=37$. Now assume that the job of client $x^T$ is scheduled last on day $x_2$ and not scheduled last on any of the days $a_2,b_2,c_2,d_2$. Then the total completion time of $x_2$ is again at most~$K=37$. Hence, we have that the job of client $x^T$ is scheduled last on day $x_2$ and w.l.o.g.\ also scheduled last on day~$a_2$. This, however, is a contradiction since the job of client $x^T$ being scheduled last on day $x_2$ means (by construction of~$\mathcal{S}$) that in the satisfying assignment $x$ is set to \texttt{false} but the job of client $x^T$ being scheduled last on day $a_2$ means (by construction of $\mathcal{S}$) that $x$ appears non-negated in clause $a_2$ and $a_2$ is satisfied by~$x$.

($\Leftarrow$): Assume that the constructed \CESabbrv instance is a yes-instance and we have a $K$-fair set of schedules $\mathcal{S}$. We construct a satisfying assignment for $\phi$ as follows.

For every variable $x$, if the job of client $x^T$ is scheduled before the job of client $x^F$ on day~$x_2$, then we set $x$ to \texttt{true}. Otherwise, we set $x$ to \texttt{false}.

Assume for the sake of contradiction that the above described assignment is not satisfying for~$\phi$. Then there is a clause $c$ that is not satisfied. Let $x,y,z$ be the variables appearing in $c$ and let them w.l.o.g.\ all appear non-negated. If $c$ is not satisfied, then we set all three variables to \texttt{false}. Now consider day $c_2$. We assume w.l.o.g.\ that the jobs of all clients that have zero processing time on day~$c_2$ are scheduled first in an arbitrary order. Note that the job of client $c$ cannot be scheduled last on day $c_2$, since then the total completion time of client $c$ would be at least~$38$, a contradiction to the assumption that $\mathcal{S}$ is a $K$-fair schedule. It follows that w.l.o.g.\ the job of client~$x^T$ is scheduled last on day $c_2$. Now consider day $x_2$. Again, we assume w.l.o.g.\ that the jobs of all clients that have zero processing time on day $x_2$ are scheduled first in an arbitrary order. Note that the job of client~$x$ cannot be scheduled last on day $x_2$, since then the total completion time of client~$x$ would be at least $38$, a contradiction to the assumption that $\mathcal{S}$ is a $K$-fair set of schedules. It follows that the job of client $x^T$ is scheduled last on day $x_2$. However, now we have that the total completion time of client $x^T$ is at least $38$, a contradiction.
\end{proof}

\subsection{Small processing times}
%\todo[inline]{I don't like the section title / structure, but I also can't think of something better right now. the previous section also deals with unarily encoded processing times.. also, \cref{thm:nfold} (the $n$-Fold IP) technically does not require unarily encoded processing times. although if processing times are huge, the problem is trivially FPT wrt.\ $K$, but if we leave the structure like this, we should discuss this.}
Theorem~\ref{thm:strNPh} above shows that \CES remains NP-hard even if the fairness threshold and all processing times are constant. The question now remains: What happens if all processing times are small and either $n$ or $q$ is small as well? In what follows we address this question. 

We first consider the case where the number $n$ of clients is small. We show via dynamic programming that \CES is solvable in polynomial time when $n=O(1)$ and all processing times are encoded in unary (polynomially bounded in the input size). Intuitively, in our dynamic program we try out all possible schedules for a certain day~$i$, and use the dynamic programming table to look up which schedules to use for days 1 to~$i-1$.

\begin{theorem}
\label{thm:unaryfewdays}
The \CES problem is polynomial-time solvable when $n=O(1)$ and all processing times are encoded in unary.
\end{theorem}
\begin{proof}
We use a dynamic program with table $T : \{0,\ldots,q\}\times \{-1,0,\ldots,K\}^n \rightarrow \{\texttt{true},\texttt{false}\}$. Intuitively, $T(i,K_1,\ldots,K_n)$ tells us whether there exists a set of schedules for days 1 to $i$ such that the total completion time of client $i$ is at most $K_j$ for all $j\in\{1,\ldots,n\}$. Formally, we define the table as follows.
\begin{align*}
T(0,K_1,\ldots,K_n) &=    \begin{cases}
      \texttt{true}, & \text{if for all } i : K_i\neq -1,\\
      \texttt{false}, & \text{otherwise}.
    \end{cases}\\
T(i,K_1,\ldots,K_n) &= \bigvee_{\sigma_i} \ T(i-1,\max\{-1,K_1-C_{i,1}\},\ldots,\max\{-1,K_n-C_{i,n}\}).
\end{align*}

Next, we show that $T(m,K,\ldots, K)=\texttt{true}$ if and only if we face a yes-instance. More specifically, we show that $T(i,K_1,\ldots,K_n)=\texttt{true}$ if and only if there is a set of schedules for the first $i$ days such that the total completion time of client $j$ on the first $i$ days is at most $K_j$ for all $j\in\{1,\ldots,n\}$. We show this by induction on $i$. For $i=0$ this is trivially true.
Now consider $T(i,K_1,\ldots,K_n)$. We look at all schedules $\sigma_i$ for day $i$. Now consider $T(i-1,\max\{-1,K_1-C_{i,1}\},\ldots,\max\{-1,K_n-C_{i,n}\})$, where $C_{i,j}$ is the completion time of client~$j$ on day $i$ when we use $\sigma_i$. 

($\Rightarrow$): If there exists a set of schedules $\mathcal{S}^\star$ for the first $i$ days such that the total completion time of client $j$ on the first $i$ days is at most $K_j$ for all $j\in\{1,\ldots,n\}$, then this set of schedules uses a particular schedule $\sigma^\star_i$ for day $i$ which is among the schedules for day $i$ we consider. Furthermore, the existence of $\mathcal{S}^\star$ implies that there is a set of schedules for the first $i-1$ days such that the total completion time of client $j$ on the first $i-1$ days is at most $K_j-C^\star_{i,j}$ for all $j\in\{1,\ldots,n\}$, where~$C^\star_{i,j}$ is the completion time of client~$j$ on day $i$ when we use $\sigma^\star_i$. Hence, we have by induction that $T(i-1,\max\{-1,K_1-C^\star_{i,1}\},\ldots,\max\{-1,K_n-C^\star_{i,n}\})=\texttt{true}$. It follows that $T(i,K_1,\ldots,K_n)=\texttt{true}$.

($\Leftarrow$): Now for the other direction assume that $T(i,K_1,\ldots,K_n)=\texttt{true}$. This means that there exists a schedule $\sigma_i$ for day $i$ such that $T(i-1,\max\{-1,K_1-C_{i,1}\},\ldots,\max\{-1,K_n-C_{i,n}\})=\texttt{true}$, where $C_{i,j}$ is the completion time of client~$j$ on day $i$ when we use $\sigma_i$. By induction we know that $T(i-1,\max\{-1,K_1-C_{i,1}\},\ldots,\max\{-1,K_n-C_{i,n}\})=\texttt{true}$ and that there exists a set of schedules for the first $i-1$ days such that the total completion time of client $j$ on the first $i-1$ days is at most $K_j-C_{i,j}$ for all $j\in\{1,\ldots,n\}$. Hence, we can combine this set of schedules with $\sigma_i$ to obtain a set of schedules for the first $i$ days such that the total completion time of client $j$ on the first $i$ days is at most $K_j$ for all $j\in\{1,\ldots,n\}$.

Finally, we analyze the running time of the algorithm. The table size is in $O(q\cdot K^n)$ and the time to fill in one entry is in $O(n!\cdot n)$. Hence, the overall running time is $O(n!\cdot K^n\cdot n\cdot q)$.
\end{proof}

Next we consider the case of a constant number of days, $q$. We do not know whether we can we get an analogous polynomial-time result of Theorem~\ref{thm:unaryfewdays} in case $q=O(1)$. However, we can exclude the possibility of an FPT algorithm with respect to $q$ (under standard parameterized complexity assumptions). For this, we present a W[1]-hardness result by constructing a parameterized reduction from the \textsc{Unary Bin Packing} problem, which is known to be W[1]-hard when parameterized by the number of bins~$b$~\cite{jansen2013bin}.

\begin{definition}
\label{def:binpack}
\textsc{Unary Bin Packing}: Given a set $I=\{1, \ldots, |I|\}$ of items with unary encoded sizes $s_j$ for $j\in I$, and $b$ bins of size $B$, is it possible to distribute the items to the bins such that no bin is overfull, i.e., the sum of the sizes of items put into the same bin does not exceed $B$?
\end{definition}

\begin{theorem}
\label{thm:W1h}
The \CESabbrv problem is W[1]-hard when parameterized by the number of days~$q$ even if all processing times are encoded in unary. 
\end{theorem}

\begin{proof}
Given an instance of \textsc{Unary Bin Packing} (where w.l.o.g.\ we assume that $b\ge 2$), we construct an instance of \CESabbrv with $n=|I|+b$ clients and $q=2b$ days as follows.
\begin{itemize}
\item For each bin $x$ we add two days $x_1$ and $x_2$ and one client $x$. We set $p_{x_1,x}=B\cdot(b^3-b^2-b)$ and $p_{x_2,x}=B\cdot b^2$.
\item For each item $j$ we create one client with $p_{x_2,j}=s_j$ for all $1\le x\le b$.
\item For each combination of a client $j$ and a day $i$ where we have not specified $p_{i,j}$ above, we set $p_{i,j}=0$.
\item We set $K=B\cdot (b^3-b+1)$.
\end{itemize}
The \CESabbrv instance can clearly be computed in polynomial time.

\emph{Correctness.} We now show that the \textsc{Unary Bin Packing} instance is a yes-instance if and only if the constructed \CESabbrv instance is a yes-instance.

($\Rightarrow$): Assume that the given \textsc{Unary Bin Packing} instance is a yes-instance. Then there is a partition of the items $I$, let it be denoted $I_1, \ldots, I_b$ such that $\sum_{j\in I_x}s_j=B$ for all $1\le x\le b$. We schedule the jobs in the following way. 

For all $1\le x\le b$, we schedule the job of client $x$ on day $x_1$. On day $x_2$ we first schedule all clients $j$ with $j\in I_x$ in an arbitrary order, then we schedule the job of client $x$, and then we schedule jobs of the clients $j$ with $j\in I\setminus I_x$ in an arbitrary order.

Note that above we specified how to schedule all jobs with non-zero processing times. All remaining jobs are scheduled at the beginning (first) on their respective days in arbitrary order. Let this schedule be called $\mathcal{S}$.

For all $1\le x\le b$ we have that the total completion time of client $x$ is $C_x(\mathcal{S})\le p_{x_1,x}+p_{x_2,x}+B=K$. For all clients $j$ with $j\in I_x$ we have that on day $x_2$ the completion time is at most $B$, whereas on days $y_2$ with $1\le y\le b$ and $y\neq x$ the completion time can be as large as~$B\cdot (b^2+b)$ (note that $\sum_{j\in I}s_j\le b\cdot B$). It follows that $C_j\le B + (b-1)\cdot B\cdot (b^2+b)=K$. 

($\Leftarrow$): Assume that the constructed \CESabbrv instance is a yes-instance and we have a $K$-fair set $\mathcal{S}$ of schedules. We now distribute the items as follows.
For $1\le x\le b$ if client $j$ for some $j\in I$ is scheduled before client $x$ on day $x_2$, we put item $j$ into bin~$x$. If this results in an item being in multiple bins, then we remove the item from all but one bin. 

Assume for contradiction that some bin $x$ is overfull, then the total completion time of client~$x$ is larger than $K$, a contradiction. Now assume for contradiction that there is an item~$j\in I$ that is not contained in any bin. This means that for all $1\le x\le b$ we have that the completion time of client~$j$ on day $x_2$ is at least $B\cdot b^2$. Hence the total completion time of client~$j$ is at least $B\cdot b^3>K$, again a contradiction.
\end{proof}

Our next result complements nicely the hardness result of Theorem~\ref{thm:W1h} above. We show that that \CES admits an  FPT algorithm when \emph{both} the number of days $q$ and the fairness threshold $K$ are taken as a parameter. Specifically, the algorithm below shows that \CES is polynomial-time solvable when both $q$ and $K$ are bounded by some constants. The case where $q=O(1)$ and all processing times are bounded by some constant remains open.  

Our algorithm is essentially a reduction to an $n$-fold ILP~\cite{de2008n,hemmecke2013n,cslovjecsek2021block,jansen2020near,EisenbrandHK18,KouteckyLO18}. Such ILPs have been used for solving a number of different scheduling problems in FPT time~\cite{knop2018scheduling,jansen2021empowering,chen2017parameterized,jansen2020approximation}. In the following, we define for some $n \in \mathbb{N}$ the $n$-fold product of two given integer matrices. Let $A_1 \in \mathbb{N}^{\ell_1 \times s}$ and $A_2 \in \mathbb{N}^{\ell_2 \times s}$ be two integer matrices with the same number of columns~$s$. The \emph{$n$-fold product} of $A = \left( \substack{A_1\\A_2}\right)$, denoted $A^{(n)}$, is defined as the matrix
\[
A^{(n)} =
\underbrace{\begin{bmatrix}
A_1 & A_1 & \ldots & A_1 \\
A_2 & 0 & \ldots & 0 \\
0 & A_2 & \ldots & 0 \\
\vdots & \vdots & \ddots & \vdots \\
0 & 0 & \ldots & A_2
\end{bmatrix}}_{n \text{ times}}.
\]
We call the submatrices composed of $s$ consecutive columns containing a full $A_1$ matrix ``blocks''. Due to \citet{cslovjecsek2021block} 
it is known that a feasible solution $x\in\mathbb{N}^{sn}$ to the system of inequalities $A^{(n)}x \geq b$ for some $b \in \mathbb{Z}^{sn}$ can be found in $2^{O(\ell_1 \ell_2^2)} (\ell_1 \ell_2\Delta)^{O(\ell_1^2 \ell_2 +\ell_2^2)}(sn)^{1+o(1)}$ time, where $\Delta=||A^{(n)}||_\infty$ is the largest absolute value of any coefficient in $A^{(n)}$.  

\begin{theorem}\label{thm:nfold}
The \CES problem is fixed-parameter tractable when parameterized by the number of days~$q$ and the fairness parameter $K$. 
\end{theorem}

\begin{proof}
We use an $n$-fold IP where, roughly speaking, we have one ``block'' in $A^{(n)}$ for every client. We use the ``local'' constraints in $A_2$ to ensure that the total completion time for each client is at most $k$ and we use the ``global'' constraints to ensure that the implicitly produced set of schedules is feasible. Note that we do not use a specific objective function. We will show that the input \CESabbrv instance is a yes-instance if and only if the constructed $n$-fold IP has a feasible solution.

Our $n$-fold IP uses the following variables:
\begin{itemize}
    \item $x^{(j)}_{i,c}$ with $j\in\{1,\ldots,n\}$, $i\in\{1,\ldots,q\}$, and $c\in\{0,\ldots,K\}$. 
    
    Intuitively, we want $x^{(j)}_{i,c}$ to be a 0-1 variable that is 1 if and only if client~$j$ has completion time $c$ on day $i$.
    \item $y^{(j)}_{i,c}$ with $j\in\{1,\ldots,n\}$, $i\in\{1,\ldots,q\}$, and $c\in\{0,\ldots,K\}$. 
    
    Intuitively, we want $y^{(j)}_{i,c}$ to be equal to $p_{i,j}$ if $x^{(j)}_{i,c}=1$ and equal to zero otherwise.
\end{itemize}
\noindent We first describe $A_2$, which contains the following local constraints:
\begin{enumerate}
    \item $\forall i : \sum_c x^{(j)}_{i,c} = 1$ 
    
    Intuitively, this set of constraints makes sure that exactly one completion time per day is picked for client $j$.
    \item $\sum_{i,c} c\cdot x^{(j)}_{i,c}\le K$
    
    Intuitively, this constraint makes sure that the total completion time of client $j$ is at most~$K$.
    \item $\forall i : \sum_c y^{(j)}_{i,c} = p_{i,j}$
    \item $\forall i,c : y^{(j)}_{i,c} - c\cdot x^{(j)}_{i,c} \le 0$
    
    Intuitively, the above two sets of constraints make sure that $y^{(j)}_{i,c}=p_{i,j}$ if $x^{(j)}_{i,c}=1$, and $y^{(j)}_{i,c}=0$ otherwise.
\end{enumerate}
\noindent Now we describe $A_1$, containing the global constraints:
\begin{itemize}
    \item $\forall i,c : \sum_{j,c'\le c}y^{(j)}_{i,c'}\le c$
    
    Intuitively, this set of constraints makes sure that on day $i$ all jobs that we want to finish before $c$ can actually be scheduled before $c$.
\end{itemize}
The number of columns $s$ in $A_1$ and $A_2$, 
the number $\ell_1$ of $A_1$ constraints, and the number $\ell_2$ of $A_2$ constraints, are each in $O(q\cdot K)$. The largest number $\Delta$ in $A_1,A_2$ is in $O(K)$. Hence by~\citet{cslovjecsek2021block}, assuming correctness of our $n$-fold IP, we get the claimed fixed-parameter tractability result. In the remainder of the proof, we show that our constructed $n$-fold IP is correct, that is, the input \CESabbrv-instance is a yes-instance if and only if the $n$-fold IP has a feasible solution.

($\Rightarrow$): Assume the input \CESabbrv-instance is a yes-instance and we have a $K$-fair set of schedules $\mathcal{S}$. We now construct a feasible solution for our $n$-fold IP. We set $x^{(j)}_{i,c}=1$ and $y^{(j)}_{i,c}=p_{i,j}$ if $C_{i,j}=c$ and $x^{(j)}_{i,c}=y^{(j)}_{i,c}=0$ otherwise. It is straightforward to see that this fulfills the first and the third local set of constraints. The second local constraint is fulfilled since $C_j\le K$ for all clients~$j$. The fourth set of constraints is fulfilled since we have $C_{i,j}\ge p_{i,j}$ for all~$i,j$. Finally, assume for contradiction that a global constraint is violated, say for $i,c$. Then we have that $\sum_{j,c'\le c}y^{(j)}_{i,c'}> c$. By construction we know that if a variable $y^{(j)}_{i,c'}$ is non-zero, then we have $y^{(j)}_{i,c'}=p_{i,j}$ and $C_{i,j}=c'$. Now consider the largest~$c'$ for which there is a client $j$ such that $y^{(j)}_{i,c'}=p_{i,j}$. We can assume w.l.o.g.\ that~$c'=c$. Let $j^\star$ be the client for which we have $C_{i,j^\star}=c$. Then for all clients $j'$ with $C_{i,j'}<c$ we have that their job is scheduled earlier than the one of $j^\star$ on day $i$. Hence, we can write $C_{i,j^\star}=p_{i,j^\star}+\sum_{j'\mid C_{i,j'}<c}p_{i,j'}$. We can rewrite this as $C_{i,j^\star} = \sum_{j\mid C_{i,j}=c'\le c} y^{(j)}_{i,c'}$. Since $y^{(j)}_{i,c}=0$ if $C_{i,j}\neq c$, we can write $C_{i,j^\star} = \sum_{j,c'\le c}y^{(j)}_{i,c'}$. Together with the assumption that  $\sum_{j,c'\le c}y^{(j)}_{i,c'}> c$ this implies that $C_{i,j^\star}>c$, a contradiction.

($\Leftarrow$): Assume we have a feasible solution for the constructed $n$-fold IP. We create a set of schedules $\mathcal{S}$ as follows. For every day $i$ and client $j$ there is exactly one $c_{i,j}$ such that $x^{(j)}_{i,c_{i,j}}=1$, otherwise the first local constraint would be violated. For day $i$, we order the jobs of the clients by their corresponding $c_{i,j}$-values. If jobs of several clients have the same $c_{i,j}$-value, we order them arbitrarily. We use this order as a schedule for day $i$, that is, if $c_{i,j} < c_{i,j'}$, then the job of client $j$ is scheduled before the job of client $j'$. We claim that this yields a $K$-fair set of schedules.

We first show that $C_{i,j}\le c_{i,j}$ for all $i,j$. Assume for contradiction that $C_{i,j}> c_{i,j}$ for some~$i,j$. Then we have that $C_{i,j}=\sum_{j'\mid c_{i,j'}\le c_{i,j}}p_{i,j'}>c_{i,j}$. Note that we have that $y^{(j)}_{i,c}=p_{i,j}$ if $x^{(j)}_{i,c}=1$ and $y^{(j)}_{i,c}=0$ otherwise. If this is not the case, then either the third or the fourth local constraint is violated. This means that we can rewrite $C_{i,j}$ as $C_{i,j}=\sum_{j'\mid c_{i,j'}\le c_{i,j}}y^{(j')}_{i,c_{i,j'}}=\sum_{j', c\le c_{i,j}}y^{(j')}_{i,c}$. However, this implies that $C_{i,j}\le c_{i,j}$, otherwise the global constraint would be violated, a contradiction. Finally, we have for every $j$ that $C_j = \sum_i C_{i,j}\le \sum_i c_{i,j} = \sum_{i,c} c\cdot x^{(j)}_{i,c}\le K$, where the last inequality follows from the assumption that the second local constraint is not violated. This finishes the proof.
\end{proof}

\section{The \WESbf Problem}
\label{sec:WES}

In the following section we focus on the waiting time objective, and specifically the \WES problem. Recall that $W_{i,j}=C_{i,j}-p_{i,j}$, and that our goal is to find $q$ schedules such that $\sum_i W_{i,j} \leq K$ for each client $j$. This objective is quite similar to the completion time objective, however as we will see, it is somewhat harder to deal with. In particular, all hardness results of \CES carry over with various modifications. However, we do not know whether \WES is polynomial time solvable in the two day case; thus, the problem is open for both $q=2$ and $q=3$. 

Below we discuss how to modify all our results from Section~\ref{sec:CES}, apart from Theorem~\ref{thm:2days}, to the \WES problem. 

\subsection{NP-hardness results}

We begin with showing that \WES is weakly NP-hard for all $q \geq 4$ and all $n \geq2$, by slightly altering the proofs of \cref{thm:fewdaysNPh} and \cref{thm:fewclientsNPh}.

\begin{theorem}
\label{thm:fewdaysNPh2}
The \WES problem is weakly NP-hard for all $q\ge 4$. 
\end{theorem}
\begin{proof}
The reduction we present here is very similar to the one presented in the proof of \cref{thm:fewdaysNPh}. 
Given an instance $I$ of \textsc{Partition} we construct an instance for the \WES problem with $n=|A|+4$ and $q=4$ in the following way. 
\begin{itemize}
\item We create two clients $x,x'$. We set $p_{1,x}=8B$, $p_{2,x}=4B$, and $p_{3,x}=p_{4,x}=0$, and we set $p_{1,x'}=6B$ and $p_{2,x'}=p_{3,x'}=p_{4,x'}=0$.
\item We create two clients $y,y'$. We set $p_{1,y}=p_{2,y}=0$, and $p_{3,y}=8B$, $p_{4,y}=4B$, and we set $p_{1,y'}=p_{2,y'}=p_{4,y'}=0$ and $p_{3,y'}=6B$.
\item For each $a_j\in A$ we create one client $j$ with $p_{2,j}=p_{4,j}=a_j$ and $p_{1,j}=p_{3,j}=0$.
    \item We set $K=7B$.
\end{itemize}
To obtain the result for more than four days, we can add additional days where all clients have jobs with processing time zero.
The described \WES instance can clearly be computed in polynomial time.

\emph{Correctness.} We now show that the \textsc{Partition} instance is a yes-instance if and only if the constructed \WES instance is a yes-instance.

($\Rightarrow$): Assume that the given \textsc{Partition} instance is a yes-instance. Then there exists a partition of $A$ into $A_1$ and $A_2$ such that $\sum_{a_j\in A_1}a_j=\sum_{a_j\in A_2}a_j=B$. Then, we schedule the jobs in the following way. 
\begin{itemize}
\item On day 1, we first schedule the job of client $x'$ and then we schedule the job of client $x$.
\item On day 2, we first schedule the jobs of clients $j$ with $a_j\in A_1$ in an arbitrary order. Then we schedule the job of client $x$. Afterwards, we schedule the jobs of clients $j$ with $a_j\in A_2$ in an arbitrary order.
\item On day 3, we first schedule the job of client $y'$ and then we schedule the job of client $y$.
\item On day 4, we first schedule the jobs of clients $j$ with $a_j\in A_2$ in an arbitrary order. Then we schedule the job of client $y$. Afterwards, we schedule the jobs of clients $j$ with $a_j\in A_1$ in an arbitrary order.
\end{itemize}

Note that above we specified how to schedule all jobs with non-zero processing times. All remaining jobs are scheduled first on their respective days in arbitrary order.

The total waiting times of clients $x'$ and $y'$ is zero. The total waiting times of clients $x$ and $y$ are each exactly $K$, since we know that $\sum_{a_j\in A_1}a_j=\sum_{a_j\in A_2}a_j = B$. Now assume that there is a client~$j'$ with $a_{j'}\in A_1$ that is scheduled last among all clients~$j$ with $a_j\in A_1$ on both days 2 and~4. The total completion time of client~$j'$ is exactly $K$, which is not greater than its waiting time. It follows that the total waiting times of all other clients~$j$ with $a_j\in A_1$ are all also at most $K$. By symmetry, we also have that the total completion time of all clients $j$ with $a_j\in A_2$ is at most $K$.

($\Leftarrow$): Assume that the constructed \WES instance is a yes-instance and we have a $K$-fair set of schedules $\mathcal{S}$. Consider day 2. For all clients $j$ whose jobs are scheduled earlier than the job of client $x$ on day 2, we put $a_j$ into $A_1$. For all remaining clients~$j$, we put $a_j$ into $A_2$. Assume for contradiction that $\sum_{a_j\in A_1}a_j\neq\sum_{a_j\in A_2}a_j$. Consider first the case that $\sum_{a_j\in A_1}a_j>\sum_{a_j\in A_2}a_j$.
Note that the job of client $x$ needs to be scheduled second on day 1, otherwise the total waiting time of client $x'$ is larger than $K$.
It follows that $W_x(\mathcal{S})=6B+\sum_{a\in A_1}a>7B=K$, a contradiction. Now consider the case that $\sum_{a\in A_1}a<\sum_{a\in A_2}a$. Assume for contradiction that all clients~$j$ with $a_j\in A_2$ have their job scheduled earlier than client~$y$ on day 4. Then by a similar argument as above we get that $W_y(\mathcal{S})>7B=K$, a contradiction. Hence we have that there is a client $j$ that has its job scheduled after client $x$ on day 2 and after client $y$ on day 4. However, then we have that $W_j(\mathcal{S})>8B>K$, a contradiction.
\end{proof}
 
\begin{corollary}
\label{cor:fewclientsNPh2}
The \WES problem is weakly NP-hard for all $n\ge 2$. 
\end{corollary}
\begin{proof}
To obtain the result, we use the same reduction as in the proof of \cref{thm:fewclientsNPh} with only one small modification. Instead of setting $K=3B$ we set $K=B$. The correctness follows from an analogous argument as in the proof of \cref{thm:fewclientsNPh}.
\end{proof}

We next alter the proof of \cref{thm:strNPh} so that it fits the $\sum_i W_{i,j}$ objective, showing that \WES is strongly NP-hard for constant processing times and fairness threshold.   

\begin{theorem}%[\appref{thm:strNPh}]
\label{thm:strNPh2}
The \WES problem is strongly NP-hard even if $K\in O(1)$.
\end{theorem}
\begin{proof}
The reduction we present here is very similar to the one presented in the proof of \cref{thm:strNPh}. 
Given an instance for the \textsc{Exact (3,4)-SAT}, we construct an instance of \WES in the following way. 
\begin{itemize}
\item For each variable $x$ that appears in $\phi$, we create four clients $x, x', x^T, x^F$ and two days~$x_1,x_2$. We set processing times $p_{x_1,x}=37$, $p_{x_2,x}=30$, $p_{x_1,x'}=35$ and $p_{x_2,x^T}=p_{x_2,x^F}=1$.
\item For each clause $c$ that appears in $\phi$, we create two clients $c,c'$ and two days $c_1,c_2$. We set $p_{c_1,c}=37$, $p_{c_2,c}=6$, and $p_{c_1,c'}=34$. Let $x,y,z$ be the three variables appearing in $c$. If $x$ appears non-negated, then we set $p_{c_2,x^T}=1$, otherwise we set $p_{c_2,x^F}=1$. Analogously for $y$ and~$z$.
\item For each combination of a client $j$ and a day $i$ where we have not specified $p_{i,j}$ above, we set $p_{i,j}=0$.
\item We set $K=36$.
\end{itemize}
This finishes the construction of the \WES instance and can clearly be done in polynomial time. 

\emph{Correctness.} We now show that $\phi$ is satisfiable if and only if the constructed \WES instance is a yes-instance.

($\Rightarrow$): Assume that $\phi$ is satisfiable and we have a satisfying assignment for the variables in~$\phi$. We construct a $K$-fair schedule as follows.
\begin{itemize}
    \item For each variable $x$ we schedule the job of client $x'$ first on day $x_1$ and we schedule the job of client $x$ second on day $x_1$. If $x$ is set to \texttt{true} in the satisfying assignment, then on day $x_2$ we schedule the jobs of client $x_T$ first, the job of client~$x$ second, and the job of client $x^F$ third. If $x$ is set to \texttt{false} in the satisfying assignment, then on day $x_2$ we schedule the jobs of client $x_F$ first, the job of client~$x$ second, and the job of client $x^T$ third.
    \item For each clause $c$ we schedule the job of client $c'$ first on day $c_1$ and the job of client $c$ second on day $c_1$. Let $x,y,z$ be the variables in $c$ and let w.l.o.g.\ $x$ be a variable that is set to a truth-value that satisfies $c$. We schedule on day~$c_2$ the jobs of client $y^T$ (resp.\ $y^F$) first, the job of client $z^T$ (reps.\ $z^F$) second, the job of client~$c$ third, and the job of client $x^T$ (resp.\ $x^F$) fourth, depending on whether the variables appear negated or non-negated in $c$ (say, $x$ appears non-negated, then we schedule~$x^T$, otherwise we schedule $x^F$).
    \item Note that above we specified how to schedule all jobs with non-zero processing times. All remaining jobs are scheduled first (that is, before any job with non-zero processing time) on their respective days in arbitrary order.
\end{itemize}
Let the above described set of schedules be called $\mathcal{S}$. Assume for contradiction that $\mathcal{S}$ is not a $K$-fair set of schedules. Then there is a client $j$ with $W_j(\mathcal{S})>K$. It is easy to observe that for all variables $x$ we have that $W_x(\mathcal{S})=K$ and $W_{x'}(\mathcal{S})=0$ and for all clauses $c$ we have that $W_c(\mathcal{S})=K$ and $W_{c'}(\mathcal{S})=0$. It follows that $j=x^T$ or $j=x^F$ for some variable $x$.

Assume w.l.o.g\ that $j=x^T$ (the case that $j=x^F$ is analogous). We know that there are four clauses $a,b,c,d$ in $\phi$ that contain variable $x$. Note that for $i\notin\{x_2,a_2,b_2,c_2,d_2\}$ we have that $p_{i,x^T}= 0$. We claim that $W_{x^T}(\mathcal{S})>K$ implies that the job of client $x^T$ is scheduled last on day $x_2$ and on at least one of the days $a_2,b_2,c_2,d_2$. Assume that the job of client~$x^T$ is scheduled first on day $x_2$. Note that the total processing time of all jobs that do not belong to client $x^T$ on each of the days $a_2,b_2,c_2,d_2$ is $8$. This means that the total waiting time of client $x^T$ is at most~$K=36$. Now assume that the job of client $x^T$ is scheduled last on day $x_2$ and not scheduled last on any of the days $a_2,b_2,c_2,d_2$. Then the total waiting time of $x_2$ is again at most~$K=36$. Hence, we have that the job of client $x^T$ is scheduled last on day $x_2$ and w.l.o.g.\ also scheduled last on day~$a_2$. This, however, is a contradiction since the job of client $x^T$ being scheduled last on day $x_2$ means (by construction of~$\mathcal{S}$) that in the satisfying assignment $x$ is set to \texttt{false} but the job of client $x^T$ being scheduled last on day $a_2$ means (by construction of $\mathcal{S}$) that $x$ appears non-negated in clause~$a_2$ and~$a_2$ is satisfied by~$x$.

($\Leftarrow$): Assume that the constructed \WES instance is a yes-instance and we have a $K$-fair set of schedules $\mathcal{S}$. We construct a satisfying assignment for $\phi$ as follows.

For every variable $x$, if the job of client $x^T$ is scheduled before the job of client $x^F$ on day~$x_2$, then we set $x$ to \texttt{true}. Otherwise, we set $x$ to \texttt{false}.

Assume for the sake of contradiction that the above described assignment is not satisfying for~$\phi$. Then there is a clause $c$ that is not satisfied. Let $x,y,z$ be the variables appearing in $c$ and let them w.l.o.g.\ all appear non-negated. If $c$ is not satisfied, then we set all three variables to \texttt{false}. Now consider day $c_2$. We assume w.l.o.g.\ that the jobs of all clients that have zero processing time on day $c_2$ are scheduled first in an arbitrary order. First, note that the job of client $c$ cannot be scheduled first on day $c_1$, since then the total waiting time of client $c'$ would be at least $37$. Now note that the job of client $c$ cannot be scheduled last on day $c_2$, since then again the total waiting time of client $c$ would be at least~$37$, a contradiction to the assumption that $\mathcal{S}$ is a $K$-fair schedule. It follows that w.l.o.g.\ the job of client~$x^T$ is scheduled last on day $c_2$. Now consider day $x_2$. Again, we assume w.l.o.g.\ that the jobs of all clients that have zero processing time on day $x_2$ are scheduled first in an arbitrary order. First, note that the job of client $x$ cannot be scheduled first on day $x_1$, since then the total waiting time of client $x'$ would be at least $37$. Now note that the job of client~$x$ cannot be scheduled last on day $x_2$, since then again the total waiting time of client~$x$ would be at least $37$, a contradiction to the assumption that $\mathcal{S}$ is a $K$-fair set of schedules. It follows that the job of client $x^T$ is scheduled last on day $x_2$. However, now we have that the total completion time of client $x^T$ is at least $37$, a contradiction.
\end{proof}

\subsection{Small processing times}
%\todo[inline]{again, i don't like the section title / the structure. \cref{thm:strNPh2} already deals with unarily encoded processing times / numbers and \cref{cor:nfold} does not require unarily encoded processing times (and here the processing times may be larger than $K$, so we dont trivially get FPT wrt.\ $K$ if the processing times are huge).}

Observe that the total waiting time $\sum_i W_{i,j}$ of a client $j$ equals  the total completion time $\sum_i C_{i,j}$ of client $j$ minus the total processing time $\sum_i p_{i,j}$ of all his jobs. Hence, the total completion time of client $j$ is at most $K_j=K+\sum_i p_{i,j}$ if and only if their total waiting time is at most $K$. It follows that we can use the same dynamic program as described in the proof of \cref{thm:unaryfewdays} and read off the result at $T[q,K_1,K_2,\ldots,K_n]$. Thus, we get the following analog result for \WES.

\begin{corollary}
\label{cor:xp}%
The \WES problem is polynomial-time solvable when $n=O(1)$ and all processing times are encoded in unary.
\end{corollary}

\begin{theorem}
\label{thm:W1h2}
The \WES problem is W[1]-hard when parameterized by the number $q$ of days even if all processing times are encoded in unary. 
\end{theorem}
\begin{proof}
The reduction we present here is very similar to the one presented in the proof of \cref{thm:W1h}. 
Given an instance of \textsc{Unary Bin Packing} (where w.l.o.g.\ we assume that $b\ge 2$), we construct an instance of \WES with $n=|I|+2b$ clients and $q=2b$ days as follows.
\begin{itemize}
    \item For each bin $x$ we add two days $x_1$ and $x_2$ and two clients $x,x'$. We set $p_{x_1,x'}=B\cdot(b^3-2b^2-1)$, and we set $p_{x_1,x}=B\cdot b^3$ and $p_{x_2,x}=B\cdot (b^2-2b+1)$.
    \item For each item $j$ we create one client with $p_{x_2,j}=s_j$ for all $1\le x\le b$.
    \item For each combination of a client $j$ and a day $i$ where we have not specified $p_{i,j}$ above, we set $p_{i,j}=0$.
    \item We set $K=B\cdot (b^3-2b^2)$.
\end{itemize}
The \WES instance can clearly be computed in polynomial time.

\emph{Correctness.} We now show that the \textsc{Unary Bin Packing} instance is a yes-instance if and only if the constructed \WES instance is a yes-instance.

($\Rightarrow$): Assume that the given \textsc{Unary Bin Packing} instance is a yes-instance. Then there is a partition of the items $I$, let it be denoted $I_1, \ldots, I_b$ such that $\sum_{j\in I_x}s_j=B$ for all $1\le x\le b$. Then we schedule the jobs in the following way. 

For all $1\le x\le b$, we schedule the job of client $x'$ first on day $x_1$ and the job of client $x$ second. On day $x_2$ we first schedule all clients $j$ with $j\in I_x$ in an arbitrary order, then we schedule the job of client $x$, and then we schedule jobs of the clients $j$ with $j\in I\setminus I_x$ in an arbitrary order.

Note that above we specified how to schedule all jobs with non-zero processing times. All remaining jobs are scheduled at the beginning (first) on their respective days in arbitrary order. Let this schedule be called $\mathcal{S}$.

For all $1\le x\le b$ we have that the total waiting time of client $x$ is $W_x(\mathcal{S})\le p_{x_1,x'}+B=K$. The total waiting time of client $x'$ is zero. For all clients $j$ with $j\in I_x$ we have that on day $x_2$ the waiting time is at most $B$, whereas on days $y_2$ with $1\le y\le b$ and $y\neq x$ the waiting time can be as large as~$B\cdot (b^2-b+1)$ (note that $\sum_{j\in I}s_j\le b\cdot B$). It follows that $W_j\le B + (b-1)\cdot B\cdot (b^2-b+1)= K$. 

($\Leftarrow$): Assume that the constructed \WES instance is a yes-instance and we have a $K$-fair set $\mathcal{S}$ of schedules. We now distribute the items as follows.
For $1\le x\le b$ if client $j$ for some $j\in I$ is scheduled before client $x$ on day $x_2$, we put item $j$ into bin~$x$. If this results in an item being in multiple bins, then we remove the item from all but one bin. 

Observe that $\mathcal{S}$ need to schedule the job of client $x'$ before the job of $x$ on day $x_1$, otherwise the total waiting time of client $x$ would be larger than $K$. Assume for contradiction that some bin~$x$ is overfull, then again the total waiting time of client~$x$ is larger than $K$, a contradiction.
Now assume for contradiction that there is an item~$j\in I$ that is not contained in any bin. This means that for all $1\le x\le b$ we have that the waiting time of client~$j$ on day $x_2$ is at least $B\cdot (b^2-2b+1)$. Hence the total waiting time of client~$j$ is at least $B\cdot (b^3-2b^2+b)>K$, again a contradiction.
\end{proof}

\begin{corollary}\label{cor:nfold}
The \WES problem
is fixed-parameter tractable when parameterized by the number of days~$q$ and the fairness parameter $K$.
\end{corollary}
\begin{proof}
To obtain the result, we modify the $n$-fold IP presented in the proof of \cref{thm:nfold}. 

First, we make the following observation: If for some client $j$ and some day $i$ we have $p_{i,j}>K$, then the job of client $j$ needs to be scheduled last on day $i$, otherwise the total waiting time of the client whose job is scheduled after the job of client $j$ on day $i$ is larger than $K$. This also implies that if we have two clients $j,j'$ and some day $i$ such that $p_{i,j}>K$ and $p_{i,j'}>K$, we have a no-instance.

For each client $j$ we define an individual completion time bound $K_j$. Consider client $j$ and let $i_1, i_2, \ldots, i_\ell$ be the days where the job of client $j$ has a processing time that is larger than $K$. By the observation above we know that the job of client $j$ needs to be scheduled last on each of these days. 
Furthermore, observe that the total completion time of a client $j$ equals the total waiting time of client $j$ plus the sum of processing times of all their jobs.
We set
\[
K_j=K + \sum_{i\notin\{i_1,i_2,\ldots, i_\ell\}}p_{i,j} - \sum_{i\in\{i_1,i_2,\ldots, i_\ell\}, j'\neq j} p_{i,j'}.
\]
The value $K_j$ now upper-bounds the total completion time of client $j$, over all days when the processing time of their job is at most $K$, assuming their total waiting time is at most $K$.

Now we modify the $n$-fold IP presented in the proof of \cref{thm:nfold}. We use it to determine how to schedule all jobs of clients where we do not already know due to the observations above that they need to be scheduled last.

We modify the first local constraint as follows.
\[
\sum_c x^{(j)}_{i,c} = 1 \text{ if } p_{i,j}\le K \text{, otherwise } \sum_c x^{(j)}_{i,c} = 0.
\]
This modified constraint makes sure we do not schedule jobs on those days we already know that the job needs to be scheduled last. Accordingly, we also need to modify the third local constraint.
\[
\sum_c y^{(j)}_{i,c} = p_{i,j} \text{ if } p_{i,j}\le K \text{, otherwise } \sum_c y^{(j)}_{i,c} = 0.
\]
We modify the second local constraint as follows.
\[
\sum_{i,c} c\cdot x^{(j)}_{i,c}\le K_j.
\]

The correctness of the modified $n$-fold IP can be shown in an analogous way as in the proof of \cref{thm:nfold}. Note that we have that $K_j\in O(K\cdot q)$ for all $j$, and hence the largest number $\Delta$ in $A_1,A_2$ is of $O(K\cdot q)$ and we obtain the claimed fixed-parameter tractability result.
\end{proof}

%\subsection{Small number of different processing times}

%\begin{corollary}
%The \WES and \LES problems are fixed-parameter tractable when parameterized by the number~$n$ of clients and the number $p_{\#}$ of different processing times.
%\end{corollary}
%\begin{proof}
%We can slightly modify the ILP given in the proof of \cref{thm:FPTfewclients} to obtain the result. We replace the constraints $\forall j\in\{1,\ldots,n\} : \sum_{c,\sigma} C_{c,j}(\sigma)\cdot x_{c,\sigma}\le K$ with constraints $\forall j\in\{1,\ldots,n\} : \sum_{c,\sigma} W_{c,j}(\sigma)\cdot x_{c,\sigma}\le K$ for the \WES problem, or with constraints $\forall j\in\{1,\ldots,n\} : \sum_{c,\sigma} L_{c,j}(\sigma)\cdot x_{c,\sigma}\le K$ for the \LES problem, where $W_{c,j}(\sigma)$ and $L_{c,j}(\sigma)$, respectively, are the waiting time and the lateness of the job belongs to client $j$ on a day with category $c$ when schedule $\sigma$ is used on that day. Note that $W_{c,j}(\sigma)$ and $L_{c,j}(\sigma)$ can be computed in polynomial time. The correctness and running time bound follow from analogous arguments as in the proof of \cref{thm:FPTfewclients}.
%\end{proof}

%%%%%%%%%%%%%%%%%%%%%%%%%%%%%%%%%%%%%%%%%%%%%%%%%%%%%%%%%%%%%%%%
%%%% Lateness L_{i,j}
%%%%%%%%%%%%%%%%%%%%%%%%%%%%%%%%%%%%%%%%%%%%%%%%%%%%%%%%%%%%%%%%

\section{The \LESbf Problem}

We next consider the \LES problem. Recall that in this problem the job of each client $j$ on day $i$ has a due date $d_{i,j}$ and a processing time $p_{i,j}$. We want to determine whether there exists a solution such that the total lateness $\sum_i L_{i,j}=\sum_i(C_{i,j}-d_{i,j})$ of each client $j$ over all $q$ days is at most $K$. Observe that if all due dates $d_{i,j}$ are equal to zero, then \LES is equivalent to the \CES problem. Thus, all hardness results for \CES immediately carry over to \LES. Unfortunately, we cannot extend many of the positive results to \LES. In fact, the only algorithm that extends to the lateness case is when all processing times are unary and the number of clients is constant. 

\begin{corollary}\label{cor:xp}
The \LES is polynomial-time solvable when $n=O(1)$ and all processing times are encoded in unary.
\end{corollary}
\begin{proof}
Observe that the total completion time of client $j$ is at most $K_j=K+\sum_i d_{i,j}$ if its total lateness is at most $K$. It follows that we can use the same dynamic program as described in the proof of \cref{thm:unaryfewdays} and have the result at $T[q,K_1,K_2,\ldots,K_n]$. 
\end{proof}

Regarding the huge processing times case, as opposed to \CES and \WES, we can also show that the problem is NP-hard even when there are only two days in scheduling horizon. 

\begin{theorem}%[\appref{thm:fewdays}]
\label{thm:twodays}
The \LES problem is weakly NP-hard for $q=2$. 
\end{theorem}

\begin{proof}
We provide a polynomial reduction from the \textsc{Partition} problem (see Definition~\ref{def:partition}). Given an instance $I$ of \textsc{Partition}, we construct an instance with $n=|A|+2$ and $q=2$ for the \LES problem.
\begin{itemize}
    \item We create one client $x$ with $p_{1,x}=3B$, $p_{2,x}=0$ and $d_{1,x}=d_{2,x}=0$.
    \item We create one client $y$ with $p_{1,y}=0$, $p_{2,y}=3B$ and $d_{1,y}=d_{2,y}=0$.
    \item For each $a_j\in A$ we create one client $j$ with $p_{i,j}=a_j$ for $i=1,2$ and $j=1,\ldots,|A|$. We set $d_{1,j}=d_{2,j}=(B+a_j)/2$.
    \item We set $K=4B$.
\end{itemize}

Note that 
\begin{itemize}
    \item $\sum_{i=1}^2 {L_{i,x}}=\sum_{i=1}^2 {C_{i,x}}$ (as $d_{1,x}=d_{2,x}=0$);
    \item $\sum_{i=1}^2 {L_{i,y}}=\sum_{i=1}^2 {C_{i,y}}$ (as $d_{1,y}=d_{2,y}=0$); and
    \item $\sum_{i=1}^2 {L_{i,j}}=\sum_{i=1}^2 {C_{i,j}}- \sum_{i=1}^2 {d_{i,j}}=\sum_{i=1}^2 {C_{i,j}}-B-a_j$ for $j=1,\ldots,|A|$.
\end{itemize}
Therefore, asking if there exists a solution with $\sum_{i=1}^2 {L_{i,j}} \leq K=4B$ for $j \in \{x,y\}\bigcup \{1,\ldots,|A|\} $ is equivalent to asking if there is a solution with 
\begin{itemize}
    \item $\sum_{i=1}^2 {C_{i,j}} \leq K'_x=K'_y=4B$ for $j \in \{x,y\}$; and
    \item $\sum_{i=1}^2 {C_{i,j}} \leq K'_j=5B+a_j$ for $\{1,\ldots,|A|\}$.
\end{itemize}

\emph{Correctness.} We first show that the \textsc{Partition} instance is a yes-instance if and only if the constructed \LES instance is a yes-instance.

($\Rightarrow$): Assume that the given \textsc{Partition} instance is a yes-instance. Then, there exists a partition of $A$ into $A_1$ and $A_2$ such that $\sum_{a_j\in A_1}a_j=\sum_{a_j\in A_2}a_j=B$. Define $\hat{A}_{1}=\{j\mid a_j\in A_1\}$ and $\hat{A}_{2}=\{j\mid a_j\in A_2\}$. 
We schedule the jobs as follows. 
\begin{itemize}
    \item On day 1, we first schedule the job of client $y$, then we schedule the jobs of the clients in $\hat{A}_1$ in an arbitrary order followed by the job of client $x$. Afterwards, we schedule the jobs of the clients in $\hat{A}_2$ in an arbitrary order.
    \item On day 2, we first schedule the job of client $x$, then we schedule the jobs of the clients in $\hat{A}_2$ followed by the job of client $y$. Afterwards, we schedule the jobs of the clients in $\hat{A}_1$. The order in which we schedule the jobs of the clients in $\hat{A}_1$ and in $\hat{A}_2$ is opposite to their ordering on day 1.
\end{itemize}

Without loss of generality, renumber the clients such that $\hat{A}_1=\{1,\ldots,|A_1|\}$ and $\hat{A}_2=\{|A_1|+1,\ldots,|A|\}$. The schedule we use for the first day is
\[
\sigma_1=(y,1,2,\ldots,|A_1|-1,|A_1|,x,|A_1|+1,|A_1|+2,\ldots,|A|-1,|A|).
\] 
Therefore, the schedule for the second day is 
\[
\sigma_2=(x,|A|,|A|-1,\ldots,|A_1|+2,|A_1|+1,y,|A_1|,|A_1|-1,\ldots,2,1). 
\]
It follows that
\[
C_{1,x}=\sum_{j\in \hat{A}_1}{p_{1,j}}+p_{1,x}=\sum_{a_j\in A_1}{a_{j}}+p_{1,x}=B+3B=4B, 
\]
and that 
\[
C_{2,y}=\sum_{j\in \hat{A}_2}{p_{2,j}}+p_{2,y}=\sum_{a_j\in A_2}{a_{j}}+p_{2,y}=B+3B=4B.
\]
Since $C_{1,y}=p_{1,y}=C_{2,x}=p_{2,x}=0$, we have that 
\[
\sum_{i=1}^2 {C_{i,x}}=\sum_{i=1}^2 {C_{i,y}}=4B=K'_x=K'_y. 
\]
For $j=1,\ldots,|A_1|$, we have that
\[
C_{1,j}=\sum_{\ell=1}^j p_{1,j}=\sum_{\ell=1}^j a_j,
\]
and 
\[
C_{2,j}=5B-\sum_{\ell=1}^{j-1} p_{2,j}=5B-\sum_{\ell=1}^{j-1} a_{j}.
\]
Therefore,
\[
\sum_{i=1}^2 {C_{i,j}}=\sum_{i=1}^2 {C_{i,j}}=5B+a_j=K'_j.
\] 

For $j=|A_1|+1,\ldots,|A|$, we have that
\[
C_{1,j}=\sum_{\ell=1}^{|A_1|} p_{1,j}+3B+\sum_{\ell=|A_1|+1}^j p_{1,j}=\sum_{a_j\in{A_1}} a_j+3B+\sum_{\ell=|A_1|+1}^j s_j=4B+\sum_{\ell=|A_1|+1}^j a_{j}.
\]
and
\[
C_{2,j}=B-\sum_{\ell=|A_1|+1}^{j-1} p_{2,j}=B-\sum_{\ell=|A_1|+1}^{j-1} a_j.
\]
Therefore,
\[
\sum_{i=1}^2 {C_{i,j}}=\sum_{i=1}^2 {C_{i,j}}=5B+a_j=K'_j.
\] 
and we have a yes answer to the constructed instance of the \LES problem.

($\Leftarrow$): Assume that the constructed \LES instance is a yes-instance and we have a $K$-fair set of schedules $\mathcal{S}$. Without loss of generality, we may assume that the job of client $y$ is scheduled first on day 1, and the job of client $x$ is scheduled first on day 2. It follows that the completion time of job $y$ on day 1 and of job $x$ on day 2 is zero. Now, let $\hat{A}_1$ be the set of clients in which their jobs are scheduled before the job of client $x$ on day 1 (excluding the job of client $y$). Moreover, let $\hat{A}_2$ be the set of clients whose jobs are scheduled before the job of client $y$ on day 2 (excluding the job of client $x$). It follows that 
\[
\sum_{i=1}^2 {C_{i,x}}=C_{1,x}=\sum_{j \in{\hat{A}_1}} {p_{1,j}}+p_{1,x}= \sum_{j \in{\hat{A}_1}} a_{j}+3B\leq K'_x=4B, 
\]
and
\[
\sum_{i=1}^2 {C_{i,y}}=C_{2,y}=\sum_{j \in{\hat{A}_2}} {p_{2,j}}+p_{2,y}= \sum_{j \in{\hat{A}_2}} a_{j}+3B\leq K'_y=4B. 
\]
Therefore,
$\sum_{j \in{\hat{A}_1}} a_{j}\leq B$, and $\sum_{j \in{\hat{A}_2}} a_{j}\leq B$.

Now, let $\Bar{A}_1=\{1,\ldots,|A|\}\setminus A_1$ and $\Bar{A}_2=\{1,\ldots,|A|\}\setminus A_2$. We next prove that there is no client that belongs both to $\Bar{A}_1$ and $\Bar{A}_2$ (i.e., $\Bar{A}_1 \cap \Bar{A}_2 = \emptyset$). By contradiction, assume that there is a client $j$ that belongs to  $\Bar{A}_1$ and to $\Bar{A}_2$. It implies that $\sum_{i=1}^2 {C_{i,j}}>p_{1,x}+p_{2,y}=6B>K'_j$, a contradiction.

The fact $\sum_{j \in{\hat{A}_1}} a_{j}\leq B$ and $\sum_{j \in{\hat{A}_2}} a_{j}\leq B$ implies that $\sum_{j \in{\Bar{A}_1}} a_{j}\geq B$ and $\sum_{j \in{\Bar{A}_2}} a_{j}\geq B$. This, together with the fact that $\Bar{A}_1 \cap \Bar{A}_2 = \emptyset$ implies that $\hat{A}_1=\Bar{A}_2$; that $\hat{A}_2=\Bar{A}_1$ and that $\sum_{j \in{\hat{A}_1}} a_{j}= \sum_{j \in{\hat{A}_2}}  a_{j}= B$. Therefore, setting $A_1=\{a_j\mid j \in{\hat{A}_1}\}$ and $A_2=\{a_j \mid j \in{\hat{A}_2}\}$ yields a yes answer to the corresponding \textsc{Partition} instance.
\end{proof}

%%%%%%%%%%%%%%%%%%%%%%%%%%%%%%%%%%%%%%%%%%%%%%%%%%%%%%%%%%%%%
%%%% Price Of Fairness
%%%%%%%%%%%%%%%%%%%%%%%%%%%%%%%%%%%%%%%%%%%%%%%%%%%%%%%%%%%

\section{The Price of Fairness}
\label{sec:POF}
%\todo[inline]{should we mention that this concept is inspired by the price of anarchy? maybe also comparing a best fair solution to an optimal solution makes sense? analogous to the price of stability..}
In this section we show that one can use our model to also analyze the price of applying a fair policy in a repetitive scheduling environment. Recall that in our example given in Section~\ref{subsec:example}, the global total completion time was maximized under an optimal fair schedule. The natural question to ask is: what is the worst ratio between the optimal global policy and the optimal fair policy? 

We use the following notation. Let $I$ be an instance of some \threefield{\alpha}{\beta,\rep}{\max_j \sum_i Z_{i,j}} problem, and let $\text{FEAS}(I)$ denote all feasible solutions (\emph{i.e.,}\ all sets of $q$ feasible schedules) for $I$. For a solution $\mathcal{S} \in \text{FEAS}(I)$, we use $Z(\mathcal{S})=\sum_i \sum_j Z_{i,j}$ as the global value of $\mathcal{S}$ under objective~$Z$. Now, let $K=K(I)$ denote the minimum fairness threshold for instance $I$, \emph{i.e.,}\ the minimum $K > 0$ such that there exists a schedule for $I$ with $\max_j \sum_i Z_{i,j} \leq K$. We use  $\text{FAIR}(I) \subseteq \text{FEAS}(I)$ to denote all feasible $K$-fair schedules of $I$. In the spirit of~\cite{BFT11}, we define the price of fairness of a repetitive scheduling problem as follows:

\begin{definition}
\label{def:POF}
The price of fairness $PoF(\Pi)$ of a \threefield{\alpha}{\beta,\rep}{\max_j {Z_{j}}} problem $\Pi$ is $\delta$, if $\delta$ is the minimum value that satisfies 
\[
\frac{\max_{\mathcal{S} \in \text{FAIR}(I)} Z(\mathcal{S})}{\min_{\mathcal{S} \in \text{FEAS}(I)} Z(\mathcal{S})} 	\leq \delta
\]
over all possible instances $I$ of $\Pi$. 
\end{definition}

Thus, the smaller the price of fairness of a problem is, the better a fair solution would be in the global sense. Note that the above definition assumes that $Z$ is a \emph{minimization criterion}; however, it is easy to alter it to apply for maximization criteria as well. Furthermore, note that the above define assumes the scheduler is oblivious to the global objective value. Thus, we consider the \emph{worst} (globally speaking) $K$-fair solution in the Definition~\ref{def:POF}. It is also natural to assume that the scheduler is mindful to the global objective value (\emph{e.g.} in case he has unlimited computational power), and thus to consider the best $K$-fair solution in the definition above. This is somewhat analogous to difference between the \emph{price of anarchy} and \emph{price of stability} in game theory~\cite{AGTBook}, where one can consider either the worst or best equilibrium in terms of global utility.

We next analyze the price of fairness of the \CES problem. In particular, we determine the price of fairness of \CES up to a factor of $2+\varepsilon$, for any~$\varepsilon > 0$.
\begin{theorem}
\label{thm:PoF}
For any $\varepsilon >0$, we have 
\[
n/(2+\varepsilon) \,\leq\, \textrm{PoF(\CES)} \, \leq \,  n.
\]
\end{theorem}

It is well known that processing jobs in an non-decreasing order of processing times (SPT order) minimizes the total completion time of a single day \threefield{1}{}{\sum C_j} instance. Thus, a feasible solution~$\mathcal{S}^\text{SPT}$ that minimizes the global total completion time of a \CES instance is a set of $q$ schedules, $\sigma^\text{SPT}_1, \ldots, \sigma^\text{SPT}_q$, where each $\sigma^\text{SPT}_i$ schedules the client jobs according to their SPT order on day~$i$. 

Let us focus on an arbitrary day $i\in\{1,\ldots,q\}$. Assume for simplicity that the clients are ordered in non-decreasing order of their job processing time on day~$i$, \emph{i.e.}\ $p_{i,1} \leq p_{i,2} \leq \ldots p_{i,n}$. Then the total competition time $C(\sigma^\text{SPT}_i) = \sum_j C_{i,j}$ of these jobs under $\sigma_i$ is
\[
C(\sigma^\text{SPT}_i) \,=\, \sum^n_{\ell=1} \ell \cdot p_{i,n+\ell-1}
\,=\, np_{i,1}+(n-1)p_{i,2}+\cdots+2p_{i,n-1}+p_{i,n}.
\]
Moreover, the worst any feasible solution can do on day $i$ is to schedule the clients jobs in an non-increasing order of processing times (LPT order). The total completion time of such a schedule~$\sigma^\text{LPT}_i$ would then be
\[
C(\sigma^\text{LPT}_i) \,=\, \sum^n_{\ell=1} \ell \cdot p_{i,\ell}
\,=\, p_{i,1}+2p_{i,2}+\cdots+(n-1)p_{i,n-1}+np_{i,n}.
\]
From the two observations above, we can get a trivial upper bound for the price of fairness of \CES.

\begin{lemma}
\label{lem:POFupperbound}
The price of fairness of \CES is at most $n$.    
\end{lemma}

\begin{proof}
Consider some arbitrary instance $I$ of \CES, and let $\mathcal{S}^\text{SPT}=\{\sigma^\text{SPT}_1,\ldots,\sigma^\text{SPT}_q\}$ and $\mathcal{S}^\text{LPT}=\{\sigma^\text{LPT}_1,\ldots,\sigma^\text{LPT}_q\}$ be the corresponding SPT and LPT solutions for $I$ as discussed above. Then we have 
\[
\frac{\max_{\mathcal{S} \in \text{FAIR}(I)} C(\mathcal{S})}{\min_{\mathcal{S} \in \text{FEAS}(I)} C(\mathcal{S})} 	\leq \frac{\max_{\mathcal{S} \in \text{FEAS}(I)} C(\mathcal{S})}{\min_{\mathcal{S} \in \text{FEAS}(I)} C(\mathcal{S})} =  \frac{C(\mathcal{S}^\text{LPT})}{C(\mathcal{S}^\text{SPT})} \leq \max_i \frac{ C(\sigma^\text{LPT}_i)}{C(\sigma^\text{SPT}_i)}.
\]    
Consider the day $i \in \{1,\ldots,q\}$ that maximizes the ratio above, and let $p_{i,1} \leq p_{i,2} \leq \cdots \leq p_{i,n}$ be the processing times of the client jobs on this day. Then we have: 
\[
\frac{C(\sigma^\text{LPT}_i)}{C(\sigma^\text{SPT}_i)} 	= 
\frac{\sum^n_{\ell=1} \ell \cdot p_{i,n-\ell+1}}{\sum^n_{\ell=1} \ell \cdot p_{i,\ell}} \leq \frac{n \cdot \sum^n_{\ell=1} p_{i,\ell}}{\sum^n_{\ell=1} p_{i,\ell}} = n,
\]    
and so the lemma follows.
\end{proof}

Thus, the upper bound of Theorem~\ref{thm:PoF} is given in Lemma~\ref{lem:POFupperbound} above. To complete the proof of the theorem, we prove a corresponding lower bound in Lemma~\ref{lem:POFlowerbound} below. 

\begin{lemma}
\label{lem:POFlowerbound}
The price of fairness of \CES is at least $n/(2+\varepsilon)$, for any $\varepsilon > 0$.     
\end{lemma}

\begin{proof}
Let $p:=p(\varepsilon,n)$ be the smallest integer such that $\varepsilon p \geq (n+2)(n-1)$. We construct a two day instance ($q=2$) of \CES with $n$ clients, where
\begin{itemize}
\item $p_{1,j}=p_{2,j}=1$ for $j \in \{1,\ldots,n-1\}$, and 
\item $p_{1,n}=p_{2,n}=p$.
\end{itemize}

Clearly the optimal global competition time schedule is the solution $\mathcal{S}^\text{FEAS}=\{\sigma^\text{SPT}_1,\sigma^\text{SPT}_2\}$ which schedules all client jobs in an SPT order on both days. On the other hand, considering the  algorithm of Theorem~\ref{thm:2days} on the instance above, we can observe that this algorithm will compute a solution $\mathcal{S}^\text{FAIR}=\{\sigma^\text{SPT}_1,\sigma^\text{LPT}_2\}$ which schedules all client jobs in an SPT order on the first day, and all client jobs in an LPT order on the second day. Note that the total completion of all jobs in a single day is $n(n-1)/2+(n-1)+p$ when scheduled according to the SPT order, while it is $n(n-1)/2+np$ when they are scheduled according to the LPT order. Thus, we have:
\[
\frac{C(\mathcal{S}^\text{FAIR})}{C(\mathcal{S}^\text{FEAS})} = 
\frac{C(\sigma^\text{SPT}_1)+C(\sigma^\text{LPT}_2)}{C(\sigma^\text{SPT}_1) + C(\sigma^\text{SPT}_2)} = 
\frac{(n+1)(n-1+p)}{(n+2)(n-1) + 2p} \geq \frac{np}{(2+\varepsilon)p} = \frac{n}{2+\varepsilon},
\]    
and so the lower bound of the theorem follows as well.
\end{proof}

\section{Summary and Future Research}
\label{sec:conclusion}%

\newcommand{\mrrb}[2]{\multirow{#1}{*}{\rotatebox[origin=c]{0}{#2}}}
\renewcommand{\arraystretch}{1.5}
\newcolumntype{Y}{>{\centering\arraybackslash}X}
\begin{table}[t]
  \setlength{\tabcolsep}{8pt}
  \centering
  \caption
  {
    Summary of our computational complexity results for the three repetitive scheduling problems \CES, \WES, and \LES.
  }
  \def\fbs{3.5cm}
  \begin{tabularx}{\textwidth}{@{}c|Y|Y|Y@{}}%
 \hline
 % \toprule
      & \multicolumn{3}{c}{Problem Variant}\\\cline{2-4} %\cmidrule{2-4}
      Parameter & \CES 	&  \WES	& \LES	\\
      \hline
%    \midrule
      \mrrb{3}{$q$} &\multicolumn{2}{c|}{weakly NP-h for $q=4$} & \mrrb{2}{weakly NP-h for $q=2$} \\
      & P-time for $q=2$ & \emph{open} &\\
       & \multicolumn{3}{c}{strongly W[1]-hard}\\
       %\midrule
       \hline
             \mrrb{2}{$n$} &\multicolumn{3}{c}{weakly NP-h for $n=2$} \\
       & \multicolumn{3}{c}{P-time for $n=O(1)$ and unary encoding}\\
       %\midrule
       \hline
       \mrrb{1}{$K$} &\multicolumn{3}{c}{strongly NP-h for $K=O(1)$} \\
       \hline
       %\midrule
       \mrrb{1}{$q+K$} &\multicolumn{2}{c|}{fixed-parameter tractable} & \emph{open} \\
    %\bottomrule	
  \hline
  \end{tabularx}
  \label{tab:results}
\end{table}

We present a new and quite natural framework to evaluate the fairness of scheduling decisions in a class of repetitive scheduling problems, extending earlier work we presented in~\cite{BKN21}. Our framework fits a set of problems involving $n$ clients, each of which assigns a single job in each of a set of $q$ independent periods (\emph{e.g.}, days). The objective is to provide a set of $q$ schedules, one for each period, such that the poorest quality of service (QoS) received by any client is minimized. The QoS for clients may have different interpretations, and is calculated separately for each client with respect to its own jobs. 

We focused on three basic single machine repetitive scheduling problems, \CES, \WES, and \LES. While our results summarized in \cref{tab:results} provide a good picture of the computational complexity of these problems, there are still some open questions remaining:
\begin{itemize}
\item What is the computational complexity of \CES for $q=3$? What is the complexity of \WES for $q\in\{2,3\}$?
\item Do any of the three problems admit a polynomial algorithm when all processing times are encoded in unary and $q=O(1)$? 
\item Do any of the three problems admit an FPT algorithm parameterized by $n$ when all processing times are encoded in unary? 
\item Does \LES admit an FPT algorithm for the combination of $q$ and $K$ as a parameter?
\end{itemize}

Of course, a more interesting and fruitful research direction would be to find further applications of repetitive scheduling problems to be considered. Natural directions include, but are not limited to: 
\begin{itemize}
\item Considering other machine models such as $m$ machines in parallel, or flow-shop environments. 
\item Exploring other objective functions such as total tardiness $\sum_{i,j} T_{i,j} = \sum_{i,j} \max\{0,L_{i,j}\}$.
\item Studying additional job constraints such as precedence constraints or release times. 
\item Designing approximation algorithms that approximate the minimum possible fairness threshold.
\end{itemize}

To conclude, we believe that our framework provides a fertile field for future research on fairness in repetitive scheduling. We hope that our paper, at the very least, will spark further discussion and research into fairness within the scheduling community, and perhaps will be a first step into understanding this topic. 

\bibliographystyle{abbrvnat}
\bibliography{bibliography}

\begin{thebibliography}{37}
\providecommand{\natexlab}[1]{#1}
\providecommand{\url}[1]{\texttt{#1}}
\expandafter\ifx\csname urlstyle\endcsname\relax
  \providecommand{\doi}[1]{doi: #1}\else
  \providecommand{\doi}{doi: \begingroup \urlstyle{rm}\Url}\fi

\bibitem[Agnetis et~al.(2019)Agnetis, Chen, Nicosia, and Pacifici]{Agnetis2019}
A.~Agnetis, B.~Chen, G.~Nicosia, and A.~Pacifici.
\newblock Price of fairness in two-agent single-machine scheduling problems.
\newblock \emph{European Journal of Operational Research}, 276:\penalty0
  79--87, 2019.

\bibitem[Ajtai et~al.(1998)Ajtai, Aspnes, Naor, Rabani, Schulman, and
  Waarts]{Ajtai}
M.~Ajtai, J.~Aspnes, M.~Naor, Y.~Rabani, L.~J. Schulman, and O.~Waarts.
\newblock Fairness in scheduling.
\newblock \emph{Journal of Algorithms}, 29\penalty0 (2):\penalty0 306--357,
  1998.

\bibitem[Ala et~al.(2022)Ala, Simic, Pamucar, and Tirkolaee]{Ala}
A.~Ala, V.~Simic, D.~Pamucar, and E.~B. Tirkolaee.
\newblock Appointment scheduling problem under fairness policy in healthcare
  services: Fuzzy ant lion optimizer.
\newblock \emph{Expert Systems with Applications}, 207:\penalty0 117949, 2022.

\bibitem[Azar and Taub(2004)]{Taub}
Y.~Azar and S.~Taub.
\newblock All-norm approximation for scheduling on identical machines.
\newblock In \emph{Scandinavian Workshop on Algorithm Theory (SWAT)}, pages
  298--310, 2004.

\bibitem[Azar et~al.(1994)Azar, Broder, and Karlin]{Azar}
Y.~Azar, A.~Z. Broder, and A.~R. Karlin.
\newblock On-line load balancing.
\newblock \emph{Theoretical Computer Science}, 130\penalty0 (1):\penalty0
  73--84, 1994.

\bibitem[Baum et~al.(2014)Baum, Bertsimas, and Kallus]{Baum}
R.~Baum, D.~Bertsimas, and N.~Kallus.
\newblock Scheduling, revenue management, and fairness in an academic-hospital
  radiology division.
\newblock \emph{Academic Radiology}, 21:\penalty0 1322--1330, 2014.

\bibitem[Bertsimas et~al.(2011)Bertsimas, Farias, and Trichakis]{BFT11}
D.~Bertsimas, V.~F. Farias, and N.~Trichakis.
\newblock The price of fairness.
\newblock \emph{Operations Research}, 59\penalty0 (1):\penalty0 17--31, 2011.

\bibitem[Bonald and Massoulié(2001)]{Bonald}
T.~Bonald and L.~Massoulié.
\newblock Impact of fairness on internet performance.
\newblock In \emph{Proceedings of the 2001 ACM SIGMETRICS international
  conference on Measurement and Modeling of Computer Systems}, pages 82--91,
  2001.

\bibitem[Chen et~al.(2017)Chen, Marx, Ye, and Zhang]{chen2017parameterized}
L.~Chen, D.~Marx, D.~Ye, and G.~Zhang.
\newblock Parameterized and approximation results for scheduling with a low
  rank processing time matrix.
\newblock In \emph{Proceedings of the 34th Symposium on Theoretical Aspects of
  Computer Science (STACS '17)}, volume~66 of \emph{LIPIcs}, pages 22:1--22:14.
  Schloss Dagstuhl-Leibniz-Zentrum fuer Informatik, 2017.

\bibitem[Cslovjecsek et~al.(2021)Cslovjecsek, Eisenbrand, Hunkenschr{\"o}der,
  Rohwedder, and Weismantel]{cslovjecsek2021block}
J.~Cslovjecsek, F.~Eisenbrand, C.~Hunkenschr{\"o}der, L.~Rohwedder, and
  R.~Weismantel.
\newblock Block-structured integer and linear programming in strongly
  polynomial and near linear time.
\newblock In \emph{Proceedings of the 2021 ACM-SIAM Symposium on Discrete
  Algorithms (SODA '21)}, pages 1666--1681. SIAM, 2021.

\bibitem[De~Loera et~al.(2008)De~Loera, Hemmecke, Onn, and Weismantel]{de2008n}
J.~A. De~Loera, R.~Hemmecke, S.~Onn, and R.~Weismantel.
\newblock {$N$}-fold integer programming.
\newblock \emph{Discrete Optimization}, 5\penalty0 (2):\penalty0 231--241,
  2008.

\bibitem[Downey and Fellows(2013)]{DowneyFellowsNew}
R.~G. Downey and M.~R. Fellows.
\newblock \emph{Fundamentals of Parameterized Complexity}.
\newblock Springer, 2013.

\bibitem[Eisenbrand et~al.(2018)Eisenbrand, Hunkenschr{\"{o}}der, and
  Klein]{EisenbrandHK18}
F.~Eisenbrand, C.~Hunkenschr{\"{o}}der, and K.~Klein.
\newblock Faster algorithms for integer programs with block structure.
\newblock In \emph{Proceedings of the 45th International Colloquium on
  Automata, Languages, and Programming ({ICALP} '18)}, volume 107 of
  \emph{LIPIcs}, pages 49:1--49:13. Schloss Dagstuhl - Leibniz-Zentrum
  f{\"{u}}r Informatik, 2018.

\bibitem[Flum and Grohe(2006)]{FlumGrohe}
J.~Flum and M.~Grohe.
\newblock \emph{Parameterized Complexity Theory}.
\newblock Springer, 2006.

\bibitem[Garey and Johnson(1979)]{GJ79}
M.~R. Garey and D.~S. Johnson.
\newblock \emph{Computers and Intractability: {A} Guide to the Theory of
  NP-Completeness}.
\newblock W. H. Freeman, 1979.

\bibitem[Garey et~al.(1976)Garey, Johnson, and Stockmeyer]{GareyJS76}
M.~R. Garey, D.~S. Johnson, and L.~J. Stockmeyer.
\newblock Some simplified np-complete graph problems.
\newblock \emph{Theor. Comput. Sci.}, 1\penalty0 (3):\penalty0 237--267, 1976.
\newblock URL \url{https://doi.org/10.1016/0304-3975(76)90059-1}.

\bibitem[Gawiejnowicz and Suwalski(2014)]{Stan}
S.~Gawiejnowicz and C.~Suwalski.
\newblock Scheduling linearly deteriorating jobs by two agents to minimize the
  weighted sum of two criteria.
\newblock \emph{Computers and Operations Research}, 52:\penalty0 135--146,
  2014.

\bibitem[Gerstl et~al.(2017)Gerstl, Mor, and Mosheiov]{Gerstl}
E.~Gerstl, B.~Mor, and G.~Mosheiov.
\newblock Scheduling with two competing agents to minimize total weighted
  earliness.
\newblock \emph{Annals of Operations Research}, 253:\penalty0 227--245, 2017.

\bibitem[Graham et~al.(1979)Graham, Lawler, Lenstra, and Kan]{Graham79}
R.~Graham, E.~Lawler, J.~Lenstra, and A.~Kan.
\newblock Optimization and approximation in deterministic sequencing and
  scheduling: a survey.
\newblock \emph{Annals of Discrete Mathematics}, 3:\penalty0 287--326, 1979.

\bibitem[Harchol-Balter et~al.(2003)Harchol-Balter, Schroeder, Bansal, and
  Agrawal]{Balter}
M.~Harchol-Balter, B.~Schroeder, N.~Bansal, and M.~Agrawal.
\newblock Size-based scheduling to improve web performance.
\newblock \emph{ACM Transactions on Computer Systems}, 21\penalty0
  (2):\penalty0 207--233, 2003.

\bibitem[Heeger et~al.(2021)Heeger, Hermelin, Mertzios, Molter, Niedermeier,
  and Shabtay]{BKN21}
K.~Heeger, D.~Hermelin, G.~B. Mertzios, H.~Molter, R.~Niedermeier, and
  D.~Shabtay.
\newblock Equitable scheduling on a single machine.
\newblock In \emph{Proceedings of the 35th AAAI Conference on Artificial
  Intelligence, AAAI '21}, pages 11818--11825. {AAAI} Press, 2021.

\bibitem[Hemmecke et~al.(2013)Hemmecke, Onn, and Romanchuk]{hemmecke2013n}
R.~Hemmecke, S.~Onn, and L.~Romanchuk.
\newblock $n$-{F}old integer programming in cubic time.
\newblock \emph{Mathematical Programming}, 137\penalty0 (1):\penalty0 325--341,
  2013.

\bibitem[Jansen et~al.(2013)Jansen, Kratsch, Marx, and
  Schlotter]{jansen2013bin}
K.~Jansen, S.~Kratsch, D.~Marx, and I.~Schlotter.
\newblock Bin packing with fixed number of bins revisited.
\newblock \emph{Journal of Computer and System Sciences}, 79\penalty0
  (1):\penalty0 39--49, 2013.

\bibitem[Jansen et~al.(2020{\natexlab{a}})Jansen, Lassota, and
  Maack]{jansen2020approximation}
K.~Jansen, A.~Lassota, and M.~Maack.
\newblock Approximation algorithms for scheduling with class constraints.
\newblock In \emph{Proceedings of the 32nd ACM Symposium on Parallelism in
  Algorithms and Architectures (SPAA '20)}, pages 349--357. {ACM},
  2020{\natexlab{a}}.

\bibitem[Jansen et~al.(2020{\natexlab{b}})Jansen, Lassota, and
  Rohwedder]{jansen2020near}
K.~Jansen, A.~Lassota, and L.~Rohwedder.
\newblock Near-linear time algorithm for $n$-fold {ILP}s via color coding.
\newblock \emph{SIAM Journal on Discrete Mathematics}, 34\penalty0
  (4):\penalty0 2282--2299, 2020{\natexlab{b}}.

\bibitem[Jansen et~al.(2021)Jansen, Klein, Maack, and
  Rau]{jansen2021empowering}
K.~Jansen, K.-M. Klein, M.~Maack, and M.~Rau.
\newblock Empowering the configuration-{IP}: new {PTAS} results for scheduling
  with setup times.
\newblock \emph{Mathematical Programming}, pages 1--35, 2021.

\bibitem[Knop and Kouteck{\'y}(2018)]{knop2018scheduling}
D.~Knop and M.~Kouteck{\'y}.
\newblock Scheduling meets $n$-fold integer programming.
\newblock \emph{Journal of Scheduling}, 21\penalty0 (5):\penalty0 493--503,
  2018.

\bibitem[Kouteck{\'{y}} et~al.(2018)Kouteck{\'{y}}, Levin, and
  Onn]{KouteckyLO18}
M.~Kouteck{\'{y}}, A.~Levin, and S.~Onn.
\newblock A parameterized strongly polynomial algorithm for block structured
  integer programs.
\newblock In \emph{Proceedings of the 45th International Colloquium on
  Automata, Languages, and Programming ({ICALP} '18)}, volume 107 of
  \emph{LIPIcs}, pages 85:1--85:14. Schloss Dagstuhl - Leibniz-Zentrum
  f{\"{u}}r Informatik, 2018.

\bibitem[Nisan et~al.(2007)Nisan, Roughgarden, Tardos, and Vazirani]{AGTBook}
N.~Nisan, T.~Roughgarden, E.~Tardos, and V.~V. Vazirani.
\newblock \emph{Algorithmic Game Theory}.
\newblock Cambridge University Press, 2007.

\bibitem[Pinedo(2016)]{pinedo2012scheduling}
M.~Pinedo.
\newblock \emph{Scheduling: Theory, Algorithms, and Systems, 5th Edition}.
\newblock Springer, 2016.

\bibitem[Pu(2021)]{Pu}
L.~Pu.
\newblock Fairness of the distribution of public medical and health resources.
\newblock \emph{Frontiers in Public Health}, 9:\penalty0 768728, 2021.

\bibitem[Samorani et~al.(2021)Samorani, Harris, Blount, Lu, and
  Santoro]{Samorani}
M.~Samorani, S.~L. Harris, L.~G. Blount, H.~Lu, and M.~A. Santoro.
\newblock Overbooked and overlooked: Machine learning and racial bias in
  medical appointment scheduling.
\newblock \emph{to appear in Manufacturing \& Service Operations Management
  (https://doi.org/10.1287/msom.2021.0999)}, 2021.

\bibitem[Sipser(1997)]{Sipser}
M.~Sipser.
\newblock \emph{Introduction to the Theory of Computation}.
\newblock PWS Publishing Co., Boston, Massachusetts, 1997.

\bibitem[Smith(1956)]{SMITHNS1956}
W.~E. Smith.
\newblock Various optimizers for single-stage production.
\newblock \emph{Naval Research Logistics Quarterly}, 3:\penalty0 59--66, 1956.

\bibitem[Tovey(1984)]{Tov84}
C.~A. Tovey.
\newblock A simplified {NP}-complete satisfiability problem.
\newblock \emph{Discrete Applied Mathematics}, 8\penalty0 (1):\penalty0 85--89,
  1984.

\bibitem[Wierman(2011)]{Wierman}
A.~Wierman.
\newblock Fairness and scheduling in single server queues.
\newblock \emph{Surveys in Operations Research and Management Science},
  16:\penalty0 39--48, 2011.

\bibitem[Yan et~al.(2015)Yan, Tang, Jiang, and Fung]{Yan}
C.~Yan, J.~Tang, B.~Jiang, and R.~Y. Fung.
\newblock Sequential appointment scheduling considering patient choice and
  service fairness.
\newblock \emph{International Journal of Production Research}, 53:\penalty0
  7376--7395, 2015.

\end{thebibliography}

%\clearpage

%\appendix{}
%\appendixProofs{}

\end{document}